\documentclass[letterpaper,11pt]{article}
\usepackage{microtype}
\usepackage[letterpaper,margin=1in]{geometry}
\usepackage[UKenglish]{babel}
\usepackage[numbers,sort&compress]{natbib}
\usepackage{color}
\usepackage{doi}
\usepackage{latexsym}
\usepackage{amsmath,amsthm,amssymb}
\usepackage{enumitem}
\usepackage{xparse}
\usepackage{xspace}
\usepackage{mathtools}
\usepackage{graphicx}
\usepackage[basic]{complexity}
\usepackage[small]{caption}
\usepackage{wrapfig}
\usepackage{booktabs}
\usepackage{dsfont}
\usepackage{subcaption}
\usepackage{algpseudocode}
\usepackage{enumitem}
\usepackage[textsize=tiny]{todonotes}
\usepackage{thmtools}
\usepackage{thm-restate}

\newcommand{\hide}[1]{}

\definecolor{darkgreen}{rgb}{0,0.5,0}
\usepackage{hyperref}
\hypersetup{
    unicode=false,          
    colorlinks=true,        
    linkcolor=blue,          
    citecolor=darkgreen,        
    filecolor=magenta,      
    urlcolor=purple           
}
\usepackage[capitalize, nameinlink]{cleveref}

\theoremstyle{plain}
\newtheorem{theorem}{Theorem}
\newtheorem{lemma}[theorem]{Lemma}

\newtheorem{corollary}[theorem]{Corollary}

\theoremstyle{definition}
\newtheorem{definition}[theorem]{Definition}

\theoremstyle{remark}
\newtheorem{remark}[theorem]{Remark}

\setlist[itemize]{label=--}
\setlist[enumerate]{label=(\arabic*),labelindent=\parindent,leftmargin=*}

\DeclarePairedDelimiter\braces{\{}{\}}
\NewDocumentCommand\set{O{}mg}{\ensuremath{\braces[#1]{#2\IfNoValueTF{#3}{}{\,:\,#3}}}}

\DeclareMathOperator{\dist}{dist}

\DeclareMathOperator{\BFS}{BFS}

\newcommand{\mytitle}[1]{\noindent\textbf{#1}}
\newcommand{\myit}[1]{(\textit{#1})}
\newcommand{\algorithmInnerMargin}{0.3cm}

\newcommand{\logstar}{\log^{*}}

\newcommand{\Tnode}{\ensuremath{$T$\text{-node}}\xspace}
\newcommand{\Tnodes}{$T$-nodes\xspace}

\newcommand{\runtimegeneral}{\textbf{Runtime:~}}
\newcommand{\runtime}{\textbf{Runtime:~}}

\newcommand{\runtimes}{\textbf{Runtime for small $\Delta$:~}}
\newcommand{\runtimel}{\textbf{Runtime for large $\Delta$:~}}
\newclass{\lcl}{LCL\xspace}
\newclass{\local}{LOCAL\xspace}
\newclass{\SLOCAL}{SLOCAL\xspace}

\newcommand{\GraphI}{G_{\uparrow i}}
\newcommand{\GraphCR}{\mathcal{C}_{DCC}}
\newcommand{\GraphR}{\mathcal{G}_{DCC}}

\newcommand{\namedref}[2]{\hyperref[#2]{#1~\ref*{#2}}}

\newenvironment{mycover}
{\list{}{\listparindent 0pt
		\itemindent    \listparindent
		\leftmargin    1cm
		\rightmargin   1cm
		\parsep        0pt}%
	\raggedright
	\item\relax}
{\endlist}

\newcommand{\myemail}[1]{\,$\cdot$\, {\small #1}}
\newcommand{\myaff}[1]{\,$\cdot$\, {\small #1}\par\medskip}

\begin{document}

	\newgeometry{margin=1in,bottom=0.2in}

	\begin{mycover}
		{\huge\bfseries\boldmath Improved Distributed
                  $\Delta$-Coloring\par}
		\bigskip
		\bigskip

		\textbf{Mohsen Ghaffari}
		\myemail{ghaffari@inf.ethz.ch}
		\myaff{ETH Z\"{u}rich, Switzerland}

		\textbf{Juho Hirvonen\footnote{Supported by ERC Grant No.\ 336495 (ACDC) and Ulla Tuominen Foundation.}}
		\myemail{ juho.hirvonen@aalto.fi}
		\myaff{Aalto University, Finland}

		\textbf{Fabian Kuhn\footnote{Supported by ERC Grant No.\ 336495 (ACDC)}}
		\myemail{kuhn@cs.uni-freiburg.de}
		\myaff{University of Freiburg, Germany}

		\textbf{Yannic Maus\footnotemark[2]}
		\myemail{yannic.maus@cs.uni-freiburg.de}
		\myaff{University of Freiburg, Germany}
	\end{mycover}

\begin{abstract}
		We present a randomized distributed algorithm that computes a $\Delta$-coloring in any non-complete graph with maximum degree $\Delta \geq 4$ in $O(\log \Delta) + 2^{O(\sqrt{\log\log n})}$ rounds, as well as a randomized algorithm that computes a $\Delta$-coloring in $O((\log \log n)^2)$ rounds when $\Delta \in [3, O(1)]$. Both these algorithms improve on an $O(\log^3 n/\log \Delta)$-round algorithm of Panconesi and Srinivasan~[STOC'1993], which has remained the state of the art for the past 25 years. Moreover, the latter algorithm gets (exponentially) closer to an $\Omega(\log\log n)$ round lower bound of Brandt et al.~[STOC'16].
	\end{abstract}

	\thispagestyle{empty}
	\setcounter{page}{0}
	\restoregeometry

\newpage


\section{Introduction and Related Work}
This paper presents faster distributed algorithms, in the \local\ model, for computing a $\Delta$-coloring of any non-clique graph with maximum degree $\Delta\geq 3$. Moreover, we also provide certain structural results on the locality of the $\Delta$-coloring problem. To formally present our results and put them in the context of the area, let us start with recalling the model.

\paragraph{The \local{} Model of distributed computing~\cite{linial92, peleg2000distributed}.} The graph is abstracted as an $n$-node network $G=(V, E)$ with maximum degree at most $\Delta$. Communications happen in synchronous rounds. Per round, each node can send one (unbounded size) message to each of its neighbors. At the end, each node should know its own part of the output, e.g., its own color.

\subsection{Background and State of the Art}
Graph coloring---assigning colors to the vertices of the graph such that no two adjacent vertices have the same color---has been a central problem in the study of distributed graph algorithms. We refer to the \emph{Distributed Graph Coloring} book by Barenboim and Elkin~\cite{barenboim2013distributed}.

Much of the focus in this area has been on computing a coloring with $\Delta+1$ colors. Notice that any graph has a $(\Delta+1)$ coloring, which can be computed via a trivial sequential greedy method: Iterate through the vertices in an arbitrary order and a node picks a color that is not used by any of its at most $\Delta$ already colored neighbors. Hence, in a sense, distributed $\Delta+1$ coloring algorithms can all be viewed as attempts at parallelizing this greedy method. We are getting a better and better understanding of the complexity of this problem, see e.g., the very recent work of Chang et al.~\cite{chang2017optimal}, which provides a $2^{O(\sqrt{\log\log n})}$-round randomized algorithm for $(\Delta+1)$-coloring, and the references therein.

On the other hand, $\Delta$-coloring is a problem of a very different nature. By a beautiful result of Brooks from 1941~\cite{brooks_1941,brooks2009colouring}, every connected graph admits a $\Delta$ coloring, unless it is exactly a complete graph or an odd cycle. The proof is of course far less trivial compared to that of $(\Delta+1)$-coloring. See the 1975 work of Lov\'{a}sz~\cite{lovasz1975three} for a simplified proof, which also supplies a polynomial-time centralized algorithm for computing a $\Delta$-coloring.

\paragraph{Why should we care about $\Delta$-coloring? General Aspects.}
One can argue that this single color of difference between $\Delta$-coloring and $(\Delta+1)$-coloring is not relevant in practice. While that is probably true, we believe that there is a strong enough theoretical interest in investigating $\Delta$-coloring. We view $\Delta$-coloring as a clean and classic graph problem which reaches just outside the problems that we understand, and thus hopefully enables us to extend our understanding of the \local{} model and to develop new algorithmic tools and techniques for it. The study of $\Delta$-coloring has previously provided theoretical insight: (1) In the existential sense, Brooks' theorem and proofs of it are widely studied and covered throughout graph theory textbooks
(see e.g., \cite[Theorem 1.4]{molloy2013coloring} and \cite[Theorem 14.4]{bondy1976graph}), while $(\Delta+1)$-coloring is usually passed over as a triviality. (2) There is a sizable literature on sequential and also parallel (PRAM) algorithms for computing $\Delta$-colorings. However, the sequential variant of $(\Delta+1)$-coloring is again ignored as being a mere triviality. Moreover, the study of $(\Delta+1)$-coloring in the PRAM model effectively stopped with the MIS algorithms of Luby~\cite{luby1986simple} and Alon et al.~\cite{alon1986fast}, which led to an $O(\log n)$-round algorithm for $(\Delta+1)$-coloring.

We also note that the relation between $(\Delta+1)$-coloring and $\Delta$-coloring is  similar to the relation between the two problems of $\Delta(1+o(1))$-coloring and $(\Delta+1)$-coloring. One can argue that practically both are equally useful. However, the former can be solved easily in $2^{O(\sqrt{\log\log n})}$ rounds using methods of Barenboim et al.~\cite{barenboim2012locality}, while there is still ongoing research on $(\Delta+1)$-coloring~\cite{harris2016distributed, chang2017optimal}, which only very recently led to a $2^{O(\sqrt{\log\log n})}$-round algorithm~\cite{chang2017optimal}.

\paragraph{Why should we care about $\Delta$-coloring? Technical Distributed Aspects.} A concrete way of pointing out the difference between the two problems of $\Delta$-coloring and $(\Delta+1)$-coloring is as follows: any partial coloring of vertices with $\Delta+1$ colors can be extended to a full coloring. However, this is not true for $\Delta$-coloring: we cannot extend any partial $\Delta$-coloring to a full coloring without changing the colors of some of the already colored vertices. This issue is one of the roots of our interest in understanding the complexity of this problem.

More concretely, many of the fast randomized algorithms for local graph problems developed over the past few years rely on the so-called \emph{shattering} technique~\cite{barenboim2012locality, ghaffari16improved, harris2016distributed, ghaffari2017splitting, chang2017optimal, FGLLL17}. In a rough sense, this method performs some randomized step which computes a partial solution such that the remaining part of the problem is made of several (disconnected) components, each of which is small, e.g., think of size $\poly(\log n)$. Then, one can solve these smaller connected components using deterministic algorithms for graphs of size $\poly(\log n)$. A crucial part here is that the partial solution is such that one can readily extend it to a full solution, in fact independently in each component, without needing to alter the already computed partial solution. The problem of $\Delta$-coloring gives us one clean local problem that reaches outside this circle. In particular, it is not clear if one can do shattering for $\Delta$-coloring, i.e., it is not clear whether there is a way of computing a partial $\Delta$-coloring such that the remaining components are small and they can be colored on their own without altering the already colored part.

Furthermore, in contrast to $(\Delta+1)$-coloring, $\Delta$-coloring has an $\omega(\log^* n)$ lower bound, even for constant-degree graphs~\cite{brandt2016LLL,chang2016exponential}. The nature of this problem is very different from $(\Delta+1)$-coloring which can be computed in $O(\log^* n)$ rounds in bounded degree graphs. Recently, in the context of lower bounds for the \emph{Lov\'asz Local Lemma problem}, Brandt et al.~\cite{brandt2016LLL} proved that $\Omega(\log\log n)$-rounds are needed by any randomized $\Delta$-coloring algorithm, even in constant-degree graphs. These results led to two problems which exhibit an exponential separation between their randomized and deterministic complexity. Sinkless orientation has an $\Omega(\log \log n)$ randomized lower bound~\cite{brandt2016LLL} and an $\Omega(\log n)$ deterministic lower bound~\cite{brandt2016LLL,chang2016exponential}, with matching randomized and deterministic upper bounds~\cite{ghaffari2017splitting}. The other problem is $\Delta$-coloring, which also has an $\Omega(\log \log n)$ randomized lower bound~\cite{brandt2016LLL} and an $\Omega(\log n)$ deterministic lower bound~\cite{chang2016exponential}; however, finding matching upper bounds has remained mostly open.

\paragraph{State of the Art for distributed $\Delta$-coloring.}
Panconesi and Srinivasan gave a randomized distributed algorithm for computing a $\Delta$-coloring in $O(\log^3 n/\log \Delta)$ rounds~\cite{panconesi1992improved, Panconesi1995}. They also provided a deterministic variant of their algorithm with complexity $O(\Delta \log^2 n)$. Recently, Aboulker et al.\ \cite{aboulker18sparse} gave a more general algorithm for $d$-list coloring graphs of maximum average degree $d$ in time $O(\Delta^4 \log^3 n)$. In the special case of trees of large enough maximum degree, Chang et al.~\cite{chang2016exponential} give an $O(\log\log n)$-round randomized algorithm for computing a $\Delta$-coloring. This, combined with their deterministic lower bound $\Omega(\log n)$ \cite{chang2016exponential}, gives an exponential separation on trees. Our algorithms establish this separation in the general bounded-degree case.


\subsection{Our Results}
 Our first result is tailored to $\Delta$-coloring constant-degree graphs.
\begin{theorem} \label{thm:mainDelta}
There is a randomized distributed algorithm that, in any non-complete graph  with maximum degree $\Delta\geq 3$, computes a $\Delta$-coloring in $O\left(\sqrt{\Delta\log\Delta}\cdot\logstar\Delta \cdot\log^2 \log n\right)$ rounds, w.h.p.\footnote{As standard, we use the phrase \emph{with high probability} (w.h.p.) to indicate that an event happens with probability $1-1/n^{c}$ for a desirably large constant $c\geq 2$}
\end{theorem}
\Cref{thm:mainDelta} immediately implies an $O\big((\log\log n)^2\big)$ round algorithm for constant-degree graphs.
\begin{corollary}\label{thm:main1} There is a randomized distributed algorithm that, in any non-complete graph $G$ with maximum degree $\Delta\in [3, O(1)]$, computes a $\Delta$-coloring of $G$ in $O((\log\log n)^2)$ rounds, w.h.p..
\end{corollary}
We comment that the condition of $\Delta\geq 3$ is necessary as $2$-coloring graphs with $\Delta=2$ needs $\Omega(n)$ rounds, even if possible, e.g., in the case of an even cycle \cite{linial92,panconesi1992improved}. The round complexity of \Cref{thm:main1} gets significantly closer to the $\Omega(\log\log n)$ round lower bound of Brandt et al.~\cite{brandt2016LLL}. Even in constant-degree graphs, the previous best known bound was the $O(\log^2 n)$-round algorithm of Panconesi and Srinivasan~\cite{panconesi1992improved, Panconesi1995}.

Our second result applies to all graphs with $\Delta\geq 4$ and improves on the $O(\log^3 n/\log \Delta)$ round complexity of Panconesi and Srinivasan~\cite{panconesi1992improved, Panconesi1995}:
\begin{theorem}\label{thm:main2} There is a randomized distributed algorithm that, in any non-complete graph $G=(V, E)$ with maximum degree $\Delta \geq 4$, computes a $\Delta$-coloring in $O(\log \Delta) + 2^{O(\sqrt{\log\log n})}$ rounds, w.h.p.
\end{theorem}



We also improve the deterministic complexity of $\Delta$-coloring for graphs with $\Delta=2^{o(\sqrt{\log n})}$.
\begin{restatable}[Deterministic $\Delta$-coloring]{theorem}{deterministicDelta} \label{thm:deterministicDelta}
  Non-clique graphs of maximum degree $\Delta \geq 3$ can deterministically be  $\Delta$-colored in $O\big(\sqrt{\Delta}\cdot\log^{-3/2}\Delta\cdot\logstar\Delta\cdot\log^2 n\big)$  rounds.
\end{restatable}
Note that \Cref{thm:deterministicDelta} is only a logarithmic factor away from the $\Omega(\log_{\Delta}n)$ deterministic lower bound of \cite{chang2016exponential} when $\Delta=O(1)$.

\subsection{Our Methods} \label{ssec:methods}

Our algorithms are based on a structural result that essentially says that either a graph is easy to $\Delta$-color locally, or it expands locally. This also yields a new proof of the distributed Brooks' Theorem by Panconesi and Srinivasan.

\begin{restatable}[Distributed Brooks' Theorem]{theorem}{distBrooks} \label{thm:local-brooks}
  Let $G$ be a graph that is not a clique with maximum degree $\Delta \geq 3$, and let $G$ be $\Delta$-colored except for one node $v$. Now $G$ can be $\Delta$-colored by recoloring the $(2\log_{\Delta-1} n)$-neighborhood of $v$ and keeping the color of all nodes outside this neighborhood unchanged.
\end{restatable}
Our algorithms are based on a \emph{layering technique}. In this technique we carefully choose a \emph{base layer} $B_0\subseteq V$ that is easy to color after everything else is colored, and layers $B_1,\ldots,B_s$ where $B_i$ consists of the nodes in distance $i$ to $B_0$. To $\Delta$-color all layers one can iteratively color the layers in \emph{reverse order} while always respecting the already fixed colors. To $\Delta$-color layer $B_i, i\neq 0$ we solve list coloring on the graph $G[B_i]$. Lists are of size $(\deg_{G[B_i]}+1)$ as each node has an uncolored neighbor on a lower index layer.
At the end layer $B_0$ is (usually) colored  with different techniques.

The best way to understand the technique is the algorithm for \Cref{thm:deterministicDelta}. There the base layer $B_0$ consists of the nodes of a ruling set of $G$ with large enough distance between the nodes. 
The ruling property of $B_0$ implies that we only need few layers to cover the whole graph and due to their large distance the nodes in $B_0$ can be colored independently with \Cref{thm:local-brooks}.

Let $H$ denote the \emph{graph of remaining nodes} after the base layer and all remaining layers have been removed. In the algorithm explained above $H$ is empty.
In our randomized  algorithms the layers do not always cover the whole graph and thus $H$ might not be empty. However, the base layer is chosen such that our structural results (cf. \Cref{ssec:structuralProperties}) show that $H$ expands: In particular, we identify all \emph{small} node-induced subgraphs that are colorable regardless of colors outside the subgraphs, compute a ruling set of these subgraphs and put their nodes in the base layer $B_0$. Then the remaining graph $H$ does not have any small subgraphs that are easy to color and our structural results show that $H$ has to expand. We leverage the expansion by randomly placing \emph{'slack'} in the graph, i.e., so called \Tnodes (each \Tnode picks two of its neighbors (non adjacent) and colors them with the same color; this introduces slack at the \Tnode as it can always find a valid color after every other node of the graph is colored). Then we use those \Tnodes as a new base layer and remove -- again with the layering technique -- all nodes that have a \Tnode close by.  Due to the expansion we can show that the probability to remain after the slack placement is $1/n^{c}$ for  constant $\Delta$; for non constant $\Delta$ a node has to be much closer to a node with slack (than in the constant degree case) to be removed and we can only bound the probability to remain by $1/\poly(\Delta)$ for a suitable polynomial. However, then standard shattering techniques (cf. \Cref{lem:shattering}) show that only \emph{small} connected components remain which we color with a similar layering technique.


\medskip

We emphasize that---to the best of our knowledge---our shattering is different from all previous shattering algorithms. Previous shattering algorithms compute a partial solution to shatter the graph into small unsolved components which are then solved to complete the partial solution. Here, the nodes in the small components are the last nodes to compute their output.  Our algorithms shatter in a fundamentally different way. We shatter the graph by removing nodes from it. The nodes in remaining components are the first to compute their output. Only afterwards we add the removed nodes to the graph and let them compute their output last. The idea of putting nodes away to be colored in the end has already been used in the deterministic coloring algorithm in \cite{barenboim2013distributed} where graphs with bounded arboricity are colored. However, we are not aware of any randomized algorithm that uses this technique.

\subsection{Outline} 

 \Cref{sec:structural} provides our  core structural results for $\Delta$-coloring and its proofs are most involved. In particular, we show that 
\begin{itemize}
\item A partial $\Delta$-coloring of a graph with a single uncolored vertex $v$ can be completed to a $\Delta$-coloring of all vertices by only recoloring the vertices in a $O(\log_{\Delta} n)$ neighborhood of $v$ (Distributed Brooks Theorem, \Cref{thm:local-brooks}),
\item a partial $\Delta$-coloring of a graph with an uncolored connected component $C$ can be completed to a $\Delta$-coloring of the whole graph without changing the colors of already colored vertices if $C$ is a so called degree choosable component,
\item graphs that do not contain small diameter degree choosable components expand quickly, ()
\item the uncolored part of a graph without small diameter degree choosable components expands quickly even if we randomly place \Tnode's. (\Cref{ssec:expansionMarking}) 
\end{itemize}
We recommend to skip \Cref{ssec:expansionMarking} when reading the paper for the first time. 

\Cref{sec:preliminaries} introduces algorithmic preliminaries and state of the art results for problems such as network decomposition, $(\deg+1)$ list coloring and ruling sets that we use as subroutines in our algorithms. 
In \Cref{sec:detColoring} we use the Distributed Brooks Theorem to provide two deterministic algorithms for $\Delta$-coloring. These algorithms already contain much of the high level structure of our randomized algorithms that we present in \Cref{sec:delta-alg}. We end in \Cref{sec:conclusion} with a conclusion.


\section{Graph Colorability and Structural Results}\label{sec:structural}

In this section we study structural properties of graphs that are \emph{not} degree-list colorable, at least locally. We will show several structural results about such graphs, which essentially tell that these graphs must expand exponentially.
This will lead to a simplified proof of the ``distributed'' Brooks' theorem due to Panconesi and Srinivasan~\cite{Panconesi1995} in \Cref{ssec:distributed-brooks}.

\subsection{Gallai-Trees and Degree Choosability} \label{ssec:dcc-definitions}
\begin{definition}[Degree-Choosability]
	A graph $G$ is \emph{degree-choosable}, if for every assignment of lists $L$, such that $|L(v)| \geq \deg(v)$ for all $v$, there exists a proper coloring of $G$ with colors from $L$.
\end{definition}
\begin{figure}
\begin{subfigure}{0.99\textwidth}
  \centering
  \includegraphics[width=0.6\textwidth]{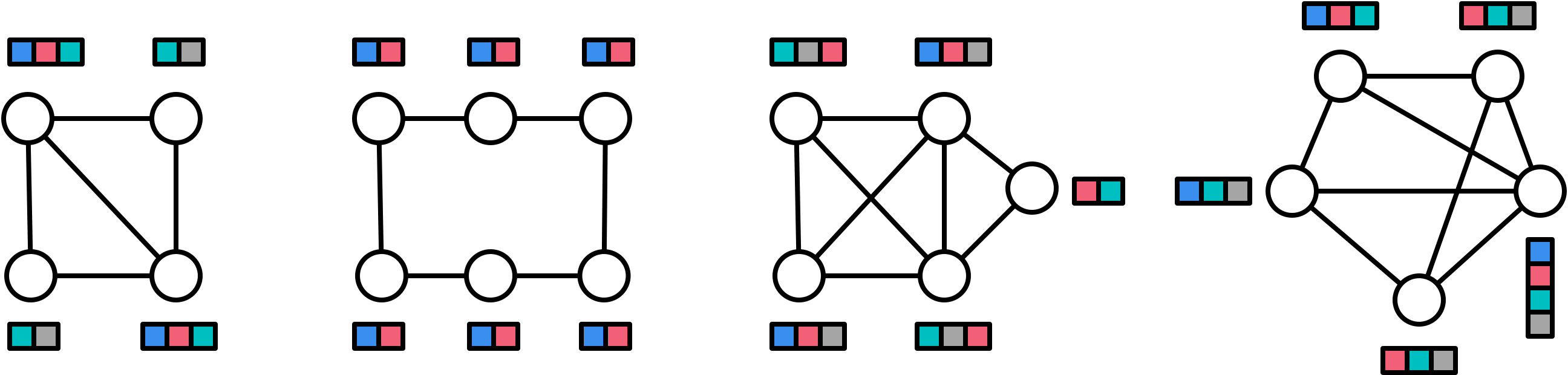}
  \label{figure:DCC1}
	\caption{Four examples of degree choosable components that can be degree list colored. }
\end{subfigure}
\begin{subfigure}{0.99\textwidth}
  \centering
  \includegraphics[width=0.6\textwidth]{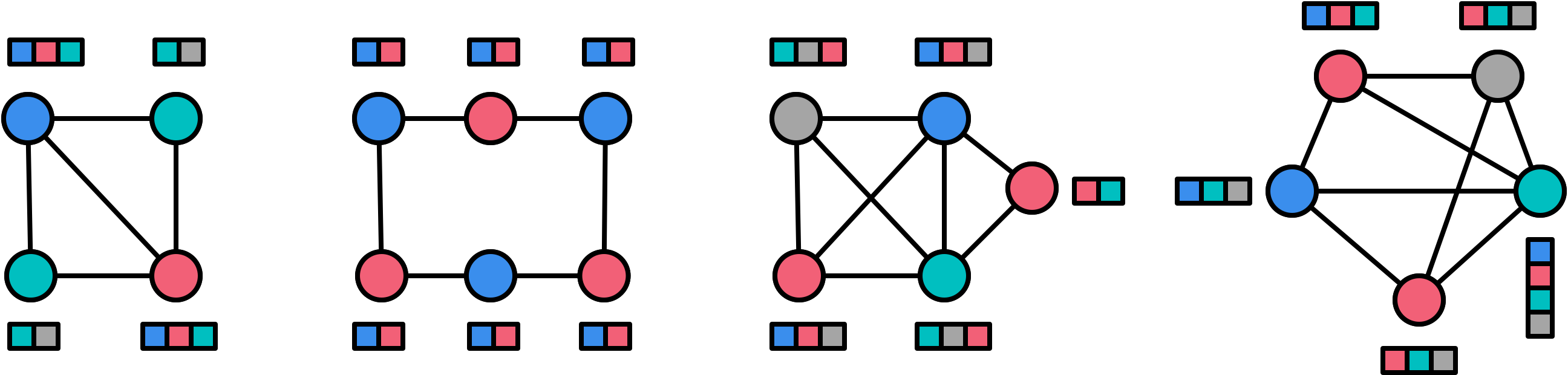}
	\caption{A valid coloring of the degree choosable components where every node has a color from its list.}
 \label{figure:DCC2}
\end{subfigure}
\caption{Degree list coloring of degree choosable graphs.}
\label{figure:DCC}
\end{figure}

\begin{definition}[Gallai-Trees]
	A graph is a \emph{Gallai-tree} if each of its maximal 2-connected components is a clique or an odd cycle.
\end{definition}

Gallai-trees are exactly those graphs which are not degree-choosable.
\begin{theorem}[\cite{erdos79choosability,vizing76vertex}] \label{thm:degree-choosablity}
	A graph is not degree-choosable if and only if it is a Gallai-tree.
\end{theorem}

Now, consider the problem of $\Delta$-coloring. Assume that we color the graph partially but leave a 2-connected subgraph that is neither a clique nor an odd cycle uncolored. Then the coloring can be completed in this subgraph due to \Cref{thm:degree-choosablity}. These $2$-connected subgraphs are called \emph{degree-choosable components} (see \Cref{figure:DCC,figure:DCC3}).
\begin{definition}[Degree-Choosable Component (DCC)]
	A node-induced subgraph is a \emph{degree-choosable component} if it is 2-connected and not a clique nor an odd cycle.
\end{definition}
We often write \emph{DCC} instead of degree choosable component and the usual graphs notions can be extended to degree-choosable components. For example, the \emph{diameter} of a degree-choosable component is the diameter of the node-induced subgraph.
A connected graph is a \emph{nice graph} if it is neither a path,
a cycle, nor a clique \cite{Panconesi1995}. Note that a degree choosable subgraph (DCC) can be $\Delta$ colored regardless of how other nodes outside the DCC are colored (see \Cref{figure:DCC3}).
All nice graphs are $\Delta$-colorable and we assume that all graphs throughout the paper are nice graphs.
A \emph{\Tnode} is a node with two neighbors that have the same color. In a partially colored graph, node $u$ is a \emph{\Tnode of $v$} if $u$ is a \Tnode and there is an uncolored path from $u$ to $v$.
For several (centralized) proofs of Brooks' Theorem and further work on Gallai trees and degree choosability we refer to \cite{cranston15brooks}.
We also want to point out that degree choosable components recently became important for the distributed coloring of sparse and planar graphs \cite{aboulker18sparse,chechik18}.

\begin{figure}
\begin{subfigure}{1.0\textwidth}
  \centering
  \includegraphics[width=0.8\textwidth]{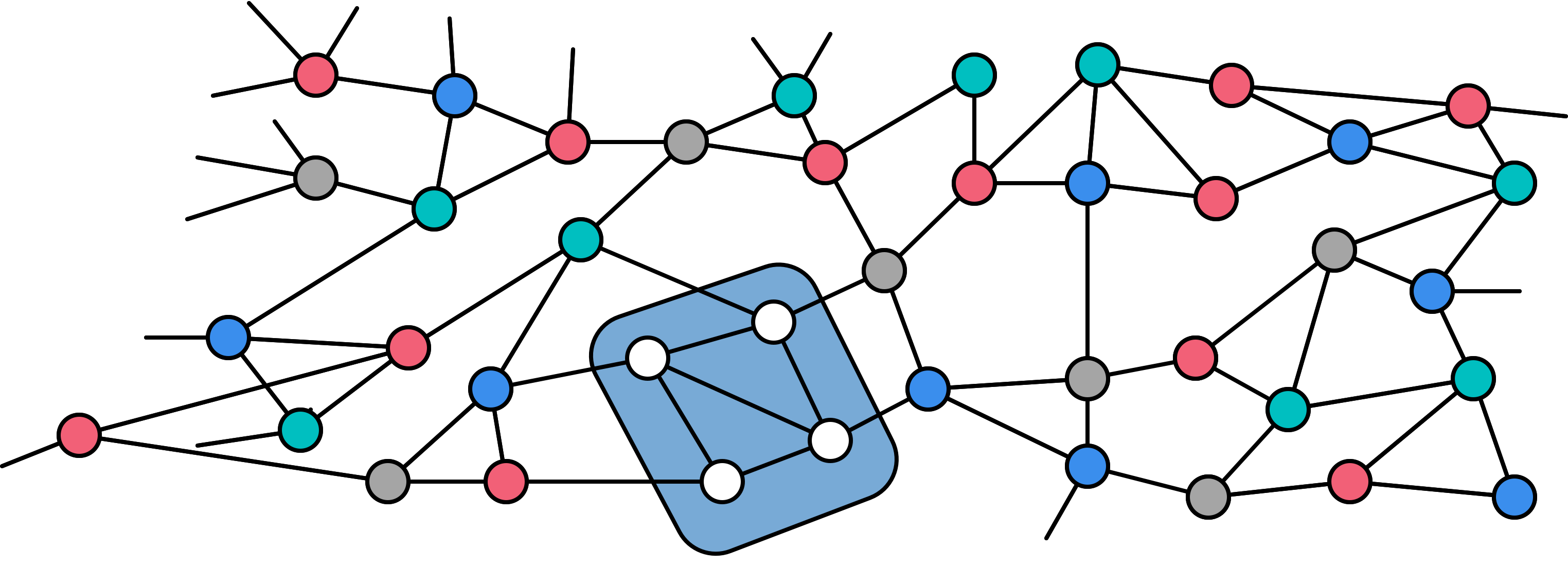}

\end{subfigure}
\caption{If a degree choosable component (DCC) appears in a graph and all nodes around it are colored then this induces a degree list coloring problem for the degree choosable component. By the definition of a DCC one can always find a valid coloring of the DCC.}
\label{figure:DCC3}
\end{figure}

\subsection{Graphs with no Small Degree-Choosable Components} \label{ssec:structuralProperties}

\begin{figure}
	\centering
	\includegraphics[width=0.4\textwidth]{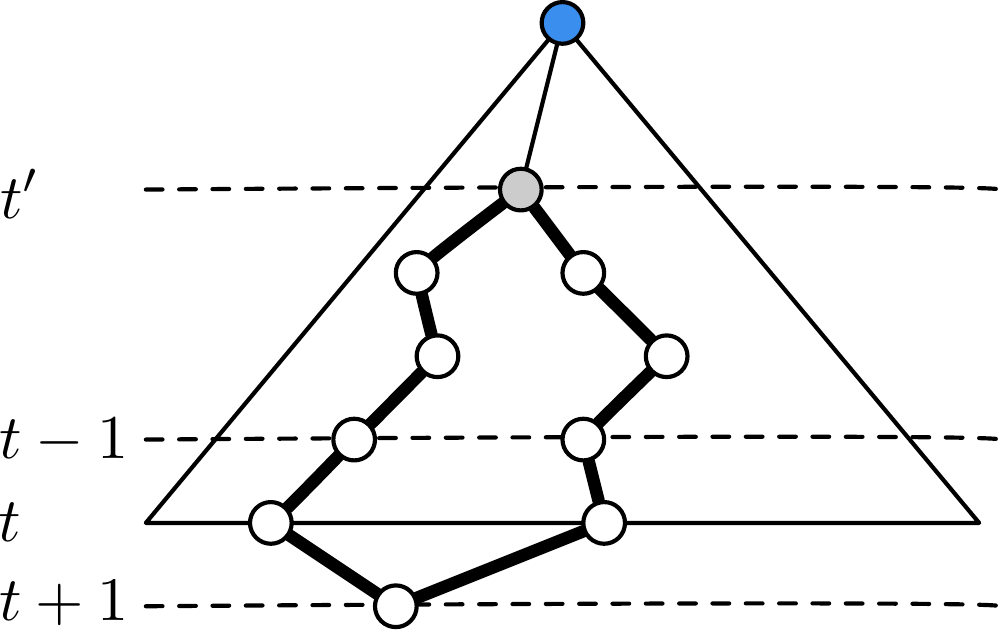}
	\caption{Uniqueness of the BFS-tree. Assuming a unique BFS-tree up to level $t$, assume that a node on the next level $t+1$ can be reached via two distinct nodes $u, v$ on level $t$. This creates an even cycle via the last common ancestor of $u$ and $v$ on some level $t' < t$. The cycle does not induce a clique since nodes of the cycle are on at least three different levels of the BFS-tree.}
	\label{fig:bfs_uniqueness}
\end{figure}

In this section we study graphs with no small degree-choosable components. Our goal is to prove that if such graphs are locally regular (and thus not easy to color locally), then these graphs must expand. Our general tool is to count the number of nodes in breadth-first search trees inside these graphs.
Given a BFS tree $\BFS(v)$ rooted at node $v$ of a graph $G$, we denote by $B_t(v)$ the set of nodes at distance $t$ from $v$ in this tree and for two nodes $u,w$ of a BFS tree let $P_{u,w}$ denote the unique path from $u$ to $w$ in the BFS tree.

\begin{lemma}[Unique BFS Tree]
\label{lemma:bfsunique}
Let $G$ be a graph with no degree-choosable components of radius $r$ or less. The depth-$r$ BFS tree rooted at an arbitrary $v\in V(G)$ is unique. In particular, any node $u\neq v$ on level $t \leq r$ has exactly one edge to the nodes on level $t-1$ of the BFS tree.
\end{lemma}
\begin{proof}[Proof of \Cref{lemma:bfsunique}]
It is immediate that the zero and depth one BFS trees are unique. For larger depth consider the following proof by contradiction.
For $t<r$ assume that $u$ and $u'$ are two nodes on the $t$-th level of the BFS tree that connect to the same node $w\notin B_{t-1}\cup B_t$, i.e., we assume that the next level of the BFS tree cannot be built uniquely. Then $w$ is on level $t+1$. Let $v'$ be the least common ancestor of $u$ and $u'$. Then there is the even cycle $\{w,u'\},P_{u',v'},P_{v',u}, \{u,w\}$ that does not induce a clique as the nodes $u,w,u'$ and $v'$ lie on three different levels, a contradiction.
\end{proof}

The following simple but useful result shows that  the neighborhood of a node decomposes into cliques  if there are no small degree-choosable components.
\begin{lemma}\label{lem:disjointCliques}
Let $G$ be a graph with no degree-choosable components of radius $1$. Then for every $v\in V(G)$ the connected components of $G[\Gamma(v)]$ are cliques.
\end{lemma}
\begin{proof}
Let $u_1, u_2$ and $w$ be neighbors of $v$ with $\{u_1,u_2\}\in E$. If $\{u_2,w\}\in E$ we also have that $\{u_1,w\}\in E$ as otherwise the graph induced by $v,u_1,u_2$ and $w$ would be a  degree choosable component of radius one, a contradiction.
\end{proof}
Note that all neighbors of any node $v$ with degree $\Delta$ have to induce more than one clique as otherwise $v$ and its neighbors would form a $(\Delta+1)$-clique and the graph does not admit a $\Delta$-coloring.

For a node $u$ in a BFS tree let $N_+(u)$ denote the set of children of $u$ in the BFS tree and let $d(u)=|N_+(u)|$~ be the number of children in the tree.

\begin{lemma}[BFS Expansion Lemma] \label{lem:deg-sum}
Let $G$ be a graph without any DCC of radius at most $r$ and $BFS(v)$ the unique depth-$r$ BFS tree rooted at some node $v\in V(G)$.  Let  $u'$ be a node of $BFS(v)$ with $\deg(u')\geq 3$ and $u$ its parent.
Then $d(u)+d(u')\geq \min\{\deg(u), \deg(u')\}$ holds.
\end{lemma}
\begin{proof}
Assume that the node $u$ is on some level $j\leq r-1$ of the BFS tree. Then  $u'$ is on level $j+1\leq r$. If $u=v$ the statement holds trivially as $d(u)=\deg(u)$. So assume that $u\neq v$.
Due to the uniqueness of the BFS tree (cf. \Cref{lemma:bfsunique}) $u\neq v$ has a single neighbor on level $j-1$ and $\deg(u)-1$ neighbors on level $j$ and $j+1$. Similarly $u$ is the only neighbor of $u'$ on level $j$.
For two nodes $u,u'$ of the BFS tree let $P(u,u')$ denote the unique path in the BFS tree between $u$ and $u'$. Let $\alpha=\min\{\deg(u),\deg(u')\}$. The result holds if $d(u)\geq \alpha$. We consider two cases for $d(u)<\alpha$.

\noindent\textit{Case $d(u)=\alpha-1$:} Assume that $u'$ has at least one neighbor in $N(u)$ that is on level $j+1$ and let $u''$ be such a neighbor of $u'$. We show that $u'$ does not have a neighbor on its level that is not connected to $u$. For contradiction, assume that $u'$ also has a neighbor $w$ on level $j+1$ that is not a child of $u$.
Let $v'$ denote the least common ancestor of $u'$ and $w$ in the BFS tree. The subgraph induced by the union of $\{u',w\}\cup \{u',u''\} \cup\{u,u''\}\cup P(v',w) \cup P(v',u)$ is a DCC of radius at most $r$, a contradiction.
Thus, if $u'$ has one neighbor in $N(u)$ it has at most $\alpha-2$ neighbors on its own level and we have $d(u')\geq \deg(u')-1-(\alpha-2)=\deg(u')-\alpha+1$. This implies $d(u)+d(u')\geq \deg(u')\geq \alpha$.

 Now assume that $u'$ has no neighbor in $N(u)$ on level $j+1$. We show that it can have at most one neighbor on its own level. Assume that it has two neighbors $w, w'$ on its own level. Let $v'$ denote the least common ancestor of $w$ and $w'$. Then the union of edges $P(v,w) \cup \{w,u'\} \cup \{u',w'\} \cup P(w',v)$ forms an even cycle that does not induce a clique, a contradiction. In this case we have $d(u')\geq \deg(u')-2\geq 1$ where the last inequality holds due to $\deg(u')\geq 3$. This implies $d(u)+d(u')\geq \alpha$.

\medskip

\noindent\textit{Case $d(u)<\alpha-1$:} Because $d(u)<\deg(u)-1$ node $u$ must have a neighbor on level $j$ that we denote by $w$. We first prove the following subclaim.

\medskip

\noindent\textbf{Subclaim:} Node $u'$ has no neighbor on level $j+1$ which is not connected to $u$.
\begin{proof}
Assume by contradiction that such a neighbor denoted with $w'$ exists. Note that the edge $\{w,w'\}$ does not exist because otherwise the nodes $u,u', w'$ and $w$ would form a 4-cycle which does not induce a clique as $\{u,w'\}\notin E$.

Let $v'$ be the least common ancestor of $w$ and $w'$.  Then \[\{u',w'\}, P_{w'v'}, P_{v'w}, \{w,u\}, \{u,u'\}\] form an even cycle.  The path $P_{w'v'}$ neither goes through $w$ nor through $u$ as the edges $\{w,w'\}$ and $\{u,w'\}$ do not exist. Thus the even cycle actually is a proper cycle and does not collapse. Furthermore, it does not induce a clique as the nodes $u,u'$ and $v'$ lie on three different levels of the BFS tree. Thus it induces a short even cycle, i.e., a DCC with radius at most $r$, a contradiction.
\renewcommand{\qedsymbol}{$\blacksquare$}
\end{proof}

Due to the subclaim all neighbors of $u'$ in $G$ that are on the same level as $u'$ in the BFS tree are also connected to $u$ in the BFS tree. As $u$ has $d(u)$ children  $u'$ has at most $d(u)-1$ neighbors on the same level and one neighbor in the level above. This implies $d(u')\geq \deg(u')-1-(d(u)-1)=\deg(u')-d(u)$. That is, $d(u)+d(u')\geq \deg(u')\geq \alpha$.
\end{proof}

Informally, this lemma holds because if there are too many edges inside the local neighborhood of a node, then these edges must create a degree-choosable component. \Cref{lem:deg-sum} implies that the BFS tree in the graph $G$ (without the marking process) expands exponentially in $\Delta-1$ with every two hops (see the proof of \Cref{lem:no-dcc-expansion} for more details).

\subsection{Exponential Expansion after the Marking Process}
\label{ssec:expansionMarking}
Next, we define a randomized marking process and show that the BFS tree of unmarked nodes either contains a \Tnode or expands exponentially as well.

\paragraph{Marking process.} In our algorithms we apply the following marking process. Each node $v$ \emph{selects} itself independently and uniformly at random with some probability $p$. Then, if there is another selected node within distance $b$ (the \emph{backoff distance}), the node unselects itself. Otherwise $v$ picks two non-adjacent neighbors and colors them with color \emph{one}. We call these neighbors \emph{marked} and say $v$ \emph{marked} them. In this case node $v$ becomes a \Tnode.
\Cref{lem:detExpansion,lem:detExpansionDelta3} show that if we apply the marking process, the graph of the unmarked nodes expands regardless of the randomness in the marking process or there was a DCC. The proof is based on the previous \Cref{lem:deg-sum}. Due to the backoff distance $b$ marked nodes cannot exist too close to each other if they are not neighbors of the same \Tnode. Thus, for every node that is blocked from expanding due to marked nodes, there are many other nodes that are not blocked on that level of the BFS tree.
These expansion results are used in the randomized algorithms for \Cref{thm:mainDelta,thm:deterministicDelta}; in particular we use the two main results of this section \Cref{lem:detExpansion} and  \Cref{lem:detExpansionDelta3}  in \Cref{sec:shatterinProbability} (to bound the failure probability in a shattering type argument) and \Cref{ssec:smallDelta} (to bound the failure probability in a union bound type argument).

We begin with a useful lemma on the structure of BFS trees in graphs without small degree choosable components. Note that the lemma can be applied to $G$ itself but also to the graph that one obtains by removing marked nodes from $G$. \Cref{lem:disjointCliques} shows that the connected components of the neighborhood of a node are cliques; the next lemma shows related but even stronger statements about the graph induced by the children in an BFS tree in such graphs. Let $w$ be a node of a BFS tree. We say that a maximal clique of children of $w$ in the BFS tree is of the \emph{odd cycle type} if it has one node that has a neighbor on the same level of the BFS tree which is not connected to $w$. We say a degree choosable component is \emph{small} if it is of radius at most $r$.

\begin{lemma}\label{lem:neighborsClique}
Let $G=(V,E)$ be a graph with maximum degree $\Delta\geq 3$ such that $G$ does not contain any DCCs of radius at most $r$ and consider an induced subgraph $H$ of $G$ and a depth-$r$ BFS tree in $H$.  Then the following hold:
\begin{enumerate}[label=(\arabic*)]
 \item If $u$ is a node in the BFS tree that has at least two neighbors on the same level of the BFS tree, then $\Delta>3$ and $u$ together with all of its neighbors on the same level of the BFS tree form a clique and all these nodes have the same parent as $u$.

\item 
If a node $u$ in the BFS tree is adjacent to a child of its parent $w$, then every neighbor of $u$ on the same level of the BFS tree is a child of $w$.
 \label{item:second}
\item Any clique of the odd cycle type consists of at most one node.
If $w$ is a node on some level $t<r$ of the BFS tree, then among its children there is at most one clique of the odd cycle type. \label{item:oddCycleType}
\end{enumerate}
\end{lemma}
\begin{proof}
\begin{enumerate}[label=(\arabic*)]
\item
Let $u_1$ and $u_2$ be two neighbors on the same level of the BFS tree and let $w$ be the parent of $u$. Further, let $v'$ be the least common ancestor of $u,u_1$ and $u_2$. If any of the edges $\{u_1,u_2\}, \{w,u_1\}, \{w,u_2\}$ does not exist then the graph induced by the edges $\{w,u\}, \{u,u_1\}, \{u,u_2\}$  and the paths $P_{v',w},P_{v',u_1}$ and $P_{v',u_2}$ implies a small DCC that is neither a clique nor an odd cycle, a contradiction.
\item Assume that $u_2$ is a second neighbor on the same level of the BFS tree and let $v'$ be the least common ancestor of $u,u_1$ and $u_2$. Assume that $u_2$ is not a neighbor of $w$. Then the nodes $u,u_1,u_2$ and the paths $P_{v',u_1}, P_{v',u_2}$ induce a small DCC, a contradiction.
\item  The first claim follows as \ref{item:second} implies that a clique with more than two nodes cannot have a neighbor with a different parent on the same level of the BFS tree.

To prove the second claim, assume for contradiction that there are two odd cycle type cliques (on level $t+1$ of the BFS tree) with parent $w$. Due to the first claim both consist of a single node that we denote by $v_i$ and $v_j$, respectively. Let $u_i$ ($u_j$) be $v_i$'s ($v_j$'s) neighbor on level $t+1$ that exist by the definition of a clique of the odd cycle type.  Then $v_i,v_j,u_i,u_j$ and the respective path to their least common ancestor form a DCC of radius at most $r$.
\qedhere
\end{enumerate}
\end{proof}

 We now show, that the BFS tree around a node after the marking process has been applied expands exponentially  in $\Delta-2$ with every two hops if we do not encounter a \Tnode.

\begin{lemma}
\label{lem:detExpansion}
Let $G=(V,E)$ be a graph with maximum degree $\Delta\geq 4$ and  such that $G$ does not contain any DCCs of radius at most $r$, for an even $r$.
Let $H$ be a graph obtained from $G$ by applying the marking process to $G$ with $b=6$ and removing all marked nodes, and let  $v \in V(H)$ be a node in $H$. If  $\deg_G(u) = \Delta$ for each $u \in N_r(v)$,
then the $r$-hop BFS tree  around node $v\in V$ has at least $(\Delta-2)^{r/2}$ nodes on level $r$ that are reachable from $v$ through paths of lengths $r$ consisting of unmarked nodes or this depth-$r$ BFS tree contains a \Tnode reachable from $v$ through a path of unmarked nodes.
\end{lemma}
\begin{proof}[Proof of \Cref{lem:detExpansion}]
Consider the BFS tree in $H$ around an arbitrary node $v$.
By Lemma~\ref{lemma:bfsunique} (applied to $H$) this BFS tree is unique and due to \Cref{lem:deg-sum} (applied to $H$) we have $d(u)+d(u') \geq \min\{\deg_H(u), \deg_H(u')\}$ for any node $u$ of the BFS tree and its child $u'$. If we encounter a node with two marked neighbors in $G$ during this process the lemma holds as the node is a \Tnode. Thus we assume that any encountered node has at most one marked neighbor.

We perform an induction on the depth of the BFS tree. The root $v$ has at least $\Delta-1$ unmarked neighbors. Let $B_t(v)$ denote the set of nodes at distance $t$ from $v$ in the BFS tree in $H$ and let $u\neq v$ be a node in $B_t(v)$.
The proof is divided into two cases based on the number of children of $u$ in the BFS tree.
\begin{enumerate}[label=(\arabic*)]
  \item \emph{$d(u) = 0$.}
	\textit{Claim: $u$ has a marked neighbor.} As $u\neq v$ it has a parent $w$ in the BFS-tree. If $u$ had no marked neighbor it would have $\Delta-1\geq 2$ neighbors on level $t$ of the BFS tree and its parent $w$ on level $t-1$. By \Cref{lem:neighborsClique} $u$, $w$ and these neighbors would  form a clique of size $\Delta+1$, a contradiction.

	Instead $u$ has a marked neighbor and only $k=\Delta-2$ neighbors on level $t$ of the BFS tree. Let $u_1,\ldots,u_k$ be these neighbors. As $k\geq 2$ \Cref{lem:neighborsClique} implies that all these neighbors form a clique and are neighbors of $w$ as well.

	\textit{Claim: None of the nodes $u_1,\ldots,u_k$ has a marked neighbor.}
		Let  $v_T$ be the \Tnode with the marked neighbors $v_{M_1}$ and $v_{M_2}$ such that $v_{M_1}$ is neighbor of $u$. Let $j\in \{1,\ldots,k\}$.
		As the distance between two $\Tnodes$ is at least $b=6$ node $u_j$ cannot have a marked neighbor other than $v_{M_1}$ or $v_{M_2}$. Any $u_j$ is not a neighbor of $v_{M_2}$ as then $v_T, v_{M_1}, v_{M_2},u,u_j,w$ induces a small DCC (the edge $\{v_{M_1}, v_{M_2}\}$ does not exist).
		If $w$ is not a neighbor of $v_{M_1}$ then none of the nodes $u_j, j=1,\ldots,k$ can have  $v_{M_1}$ as a neighbor because otherwise $w,v_{M_1}, u, u_j$ would induce a small DCC.
		If $w$ is a neighbor of $v_{M_1}$ then $v_{M_1}$ is on the same level of the BFS tree as $u$ and the other nodes and, by \Cref{lem:neighborsClique}, they all form a ($\Delta+1)$ clique with $w$, a contradiction.

 Due to the claim each $u_j, j=1,\ldots,k$ does not have a marked neighbor. We deduce that $d(u_j)=1$ holds for $j=1\ldots,k$ as $u_j$ only has $\Delta-2$ neighbors on level $t$, i.e., $u, u_1,\ldots,u_{j-1}, u_{j+1},\ldots u_k$,  (any further neighbor on this level would also be part of the clique, cf. \Cref{lem:neighborsClique}) and $w$ is the only neighbor of $u_j$ on level $t-1$. Let $u_j'$ denote the single unmarked neighbor of $u_j$ on level $t+1$ of the BFS tree.

\textit{Claim: Node $u'_j$ cannot have a marked neighbor.} With the notation from before $u$ has the marked neighbor $v_{M_1}$. Node $u_j'$ cannot have a marked neighbor that was created by a $\Tnode$ other than $v_T$ due to $b=6$. If $u'_j$ was connected to $u_{M_2}$ the nodes $v_T, v_{M_1}, v_{M_2},u,u_j, u'_j$ would form a small DCC. If $u_j'$ was connected to $u_{M_1}$ the nodes $u,u_{M_1}, u,j,u'_j$ would form a small DCC.

\smallskip

We have $\deg_H(u_j)=\Delta$ and $\deg_H(u_j')=\Delta$ and $d(u_j)=1$ for all $j=1,\ldots,k$.
  Now, we apply \Cref{lem:deg-sum} to each $u_j$ and each of their children $u_j'$ and obtain that $d(u_j') \geq \Delta - d(u_j) = \Delta-1$, i.e., $u_j'$ has $\Delta-1$ unmarked neighbors on level $t+2$.

	Therefore these nodes ($u$ and its neighbors on the same level of the BFS tree) contribute a total of at least $(\Delta-1)(\Delta-2)$ nodes to the BFS tree at level $t+2$.

	\item \emph{$d(u) \geq 1.$} There are two subcases based on whether the children of $u$ form a clique of size $\Delta-1$ or not; here we consider the topology of the children of $u$ including the children of $u$ that potentially are marked.

  First assume that they do form a clique of size $\Delta-1$. Each node in the clique has $\Delta-2$ nodes on the same level of the BFS tree, one parent and one further unique neighbor (on the next level) that potentially is marked.  Due to the uniqueness of the BFS tree (cf. \Cref{lemma:bfsunique})  all these further neighbors are distinct. At most one of the nodes of the clique can have its subtree blocked because it is marked itself or because its unique child is marked. In any case there are at least $\Delta-2$ nodes in the clique that each have one distinct unmarked child on level $t+2$ of the BFS tree.

  Now, assume that the children of $u$ do not form a $(\Delta-1)$-clique (this paragraph is still considering the topology including children of $u$ that potentially are marked). Instead, let $C_1, C_2, \dots, C_k$ form the maximal cliques formed by the children of $u$ (cf. \Cref{lem:disjointCliques}). Recall that a clique is of the odd cycle type if it is of size one and closes an odd cycle, i.e., it has $1$ edge to another node on the same level that is not a neighbor of its parent $u$. Also recall that if $|C_i| \geq 2$, there are no edges from $C_i$ to other nodes on the same level of the BFS tree (cf. \Cref{lem:neighborsClique}).

Let $C_i$ be a clique that is not of the odd cycle type. Then the only edges of node of the clique $C_i$ to nodes on level $t+1$ are the $|C_i|-1$ edges inside $C_i$. Furthermore, each vertex of the clique has one edge to  level $t$ (to $u$). Thus we obtain that any $u'$ in such a clique has $d(u') \geq \deg_H(u')-(|C_i|-1) - 1\geq \Delta-1 -(|C_i|-1) - 1 = \Delta-|C_i|-1$. Thus any such clique contributes
\[
    \sum_{u' \in C_i} d(u') \geq |C_i|(\Delta-|C_i|-1)
  \]
unmarked nodes on level $t+2$ of the BFS tree. This value is always greater than $\Delta-2$. Thus as soon as there is a clique that is not of the odd cycle type there are at least $\Delta-2$ unmarked nodes on level $t+2$ of the BFS tree. Now, assume that all cliques are of the odd cycle type.
Then, there is actually just one clique due to \Cref{lem:neighborsClique}, \ref{item:oddCycleType}.
Thus $d(u)=1$ holds. Let $u'$ be the only node of the single clique of odd cycle type. \Cref{lem:deg-sum} implies that $d(u')\geq \min\{\deg_H(u'),d_H(u)\}-d(u)\geq \Delta-2$, i.e., $u'$  has at least $\Delta-2$ unmarked children on level $t+2$ of the BFS tree.
\end{enumerate}
To bound the total number of nodes in $B_{t+2}(v)$ let $x$ denote the number of nodes $u\in B_t(v)$ of the first type, that is, with $d(u)=0$. The analysis for those nodes contributes $(\Delta-1)(\Delta-2)$ nodes on level $t+2$ of the BFS tree but also uses the other $(\Delta-2)$ nodes on level $t$ of the BFS tree that form a clique with $u$.  Thus there are at least $|B_t(v)|-x(\Delta-1)$ nodes $\tilde{u}\in B_t(v)$ of the second type, that is, nodes with $d(\tilde{u})\geq 1$, for which we can independently count their contribution of at least $(\Delta-2)$ nodes on level $t+2$ per node of the second type. Thus the total number of nodes on level $t+2$ can be bounded by
\[
|B_{t+2}(v)| \geq x(\Delta-1)(\Delta-2) + (|B_t(v)|-x(\Delta-1))(\Delta-2) = (\Delta-2)|B_t(v)|.\qedhere
\]
\end{proof}

If $\Delta=3$ an expansion that is exponential in $\Delta-2=1$ is not useful and we cannot guarantee that we expand by a factor of $\Delta-1$ every two hops after the marking process has been applied. Instead we prove that we expand by a factor of $\Delta-1$ every six hops or encounter a $\Tnode$. We first need a simple but useful lemma.

\begin{lemma}\label{lem:duOne}
Let $\Delta\geq 3$ and consider a depth-$r$ BFS tree with root $v$ in some graph without DCCs of radius at most $r$. Then the following hold:
\begin{enumerate}[label=(\arabic*)]
\item Any node of the BFS tree at distance at most $r-1$ from the root and with degree $\Delta$ in $G$ has at least one child in the BFS tree. \label{lem:oneOut}
\item Let $u$ be a node of the BFS tree and $k+dist(v,u)\leq r$. If all nodes of the $k$-hop subtree rooted at $u$ have degree $\Delta$ and are unmarked then the subtree contains at least $(\Delta-1)^{\lfloor k/2\rfloor}$ unmarked nodes on level $k$ and all of them are reachable from $u$ through paths of length at most $k$ consisting of unmarked nodes. \label{lem:noDieOut}
\end{enumerate}
\end{lemma}
\begin{proof}
\begin{enumerate}[label=(\arabic*)]
\item This is trivially satisfied for the root of the BFS tree. Let $u$ be a node that is not the root of the BFS tree and let $w$ be its parent in the tree.
To prove the claim we only need to show that $u$ cannot have $\Delta-1$ neighbors on the same level of the BFS tree. For contradiction, assume that it has neighbors $u_1,\ldots,u_{\Delta-1}$ on the same level of the tree. Due to \Cref{lem:neighborsClique} all these nodes must be also connected to $w$ which implies that $u,w,u_1,\ldots,u_{\Delta-1}$ form a $(\Delta+1)$-clique, a contradiction.

\item If $u=v$ the result follows with \Cref{lem:detExpansion}.
If $u\neq v$ let $B_t(u)$ be the nodes on level $t$ of the BFS tree rooted at $u$. Due to \ref{lem:oneOut} we have $1\leq d(w)\leq \Delta-1$ for all nodes the BFS tree rooted at $u$. The statement holds trivially for $k=0$ and we show the result by induction on $k$.
Now consider some level $k$. For each node $w \in B_k(u)$, we get via \Cref{lem:deg-sum} (we can apply the lemma as there are no marked nodes) that the number of descendants of $w$ in $B_{k+2}(u)$ is at least
	\[
		\sum_{w' \in N_+(w)} d(w') \geq d(w)(\Delta-d(w)) \geq \Delta-1.
	\]
	Since each node has a unique ancestor in the BFS tree, we get that
	\begin{align*}|B_{k+2}(v)| \geq (\Delta-1)|B_k(v)|~.& \qedhere \end{align*}
\end{enumerate}
\end{proof}

\begin{lemma} \label{lem:detExpansionDelta3}
Let $\Delta = 3$ and $G=(V,E)$ be a graph such that $G$ does not contain any DCCs of radius at most $r$, for an $r$ divisible by 6. Apply the marking process to $G$ with $b=15$. Let $v\in V$ be an arbitrary unmarked node. If $\deg_G(u) = \Delta$ for each $u \in N_r(v)$ then the $r$-hop connected component of unmarked nodes around $v$ contains a $\Tnode$ or at least  $4^{r/6}$ nodes.
\end{lemma}
\newcommand{\tree}{\ensuremath{\mathfrak{T}}\xspace}
\begin{proof}
We consider the tree $\tree$ around $v$ that one obtains by cutting the BFS tree around $v$ in $G$ whenever a marked node is encountered.
Let $\tree_t(v)$ denote the nodes on level $t$ of \tree. We show by induction on $t$  that $\tree_t(v)$ contains at least $4^{t/6}$ nodes if we do not encounter a \Tnode (note that the marked nodes that make a node of $\tree_t(v)$ a \Tnode are not part of $\tree_t(v)$). If we encounter a node with two marked neighbors during the induction the lemma holds as it forms a \Tnode. Thus we assume that any encountered node has at most one marked neighbor.

\smallskip

\emph{Base case: There are at least $4$ unmarked nodes on level 6 of \tree.} Due to $b=15$ at most two of the $5$-hop subtrees rooted at the children of $v$ are cut off due to a  marked node. The remaining $5$-hop subtree rooted at the third child cannot be cut off due to a marked node. Thus it contributes at least  $(\Delta-1)^2=4$ (unmarked) nodes on level $6$ of the tree \tree due to \Cref{lem:duOne}, \ref{lem:noDieOut}.

\smallskip

\emph{Inductive step.} Consider the set $\tree_t(v)$ for some $t \geq 6$. We split the proof in three cases: for each node $u \in \tree_t(v)$ either $d(u) = 2$, $d(u)=1$, or $d(u)=0$ holds where $d(u)$ denotes the number of children of $u$ in the tree $\tree$.

\smallskip

\emph{Case $d(u) = 2$.} As $u$ has a parent and two children, say $u_1$ and $u_2$, in \tree it does not have a marked neighbor. Due to $b = 15$ there can be at most two marked nodes in the $6$-hop subtree rooted at $u$ and both stem from the same \Tnode. If the $5$-hop subtree of any of $u$'s children has no marked nodes it contributes at least $4$ unmarked nodes to $\tree_{t+6}$ due to \Cref{lem:duOne}, \ref{lem:noDieOut}. If both $5$-hop subtrees rooted at $u_1$ and $u_2$ have a marked node, say $u_{M_1}$ and $u_{M_2}$, then they stem from the same $\Tnode$ $v_T$. Then $\{v_T,v_{M_1}\}, \{v_T,v_{M_2}\}, P_{u,v_{M_1}}, P_{u,v_{M_2}}$ must induce an odd cycle as otherwise it induces a small DCC, in particular, the edge $\{u_1,u_2\}$ does not exist. Thus $u_1$ has a neighbor $u_1'$ that lies not on the path $ P_{u,v_{M_1}}$ and the $4$-hop subtree rooted at $u_1'$ does not contain a marked node. Further, note that the $4$-hop subtree rooted at $u_1'$ is part of the subtree rooted at $u$ as $u_1'$ must be a child of $u_1$ in \tree: Assume $u'_1$ is not a child of $u_1$, that is it is on level $t+1$ of the BFS tree. Then $\{v_T,v_{M_1}\}, \{v_T,v_{M_2}\}, P_{u,v_{M_1}}, P_{u,v_{M_2}}$ and $P_{v',u'_1}, P_{v',u}, \{u_1,u'_1\}$ induce a small DCC where $v'$ is the least common ancestor of $u$ and $u'_1$.
Hence \Cref{lem:duOne}, \ref{lem:noDieOut} implies that the $4$-hop subtree rooted at $u_1'$ contributes at least $4$ unmarked nodes to $\tree_{t+6}$.

\smallskip

\emph{Case $d(u) = 1$.} First consider the case that $u$ has a marked neighbor: Let $u'$ be the single child of $u$ in \tree. If the $5$-hop subtree rooted at $u'$ does not contain a marked node it contributes at least $4$ nodes to $\tree_{t+6}(v)$ due to \Cref{lem:duOne}, \ref{lem:noDieOut}. So assume that it contains a marked node $v_{M_2}$ and let $v_{M_1}$ be the marked neighbor of $u$. Note that $v_{M_1}\neq v_{M_2}$ as the subtree rooted at $u'$ does not contain $v_{M_1}$. Then $u'$ cannot be a neighbor of $v_{M_1}$ as otherwise $D:=\{u,v_{M_1}\}, \{v_{M_1},v_{T}\}, \{v_T,v_{M_2}\}, P_{u',v_{M_2}}$ induces a small DCC. $u'$ can also not have another neighbor $u''$ on level $t+1$ of the BFS tree as then $D$ together with $\{u',u''\}, P_{v',u}, P_{v',u''}$ induces a small DCC where $v'$ is the least common ancestor of $u$ and $u''$. Thus $u'$ has a further child $w$ that does not lie on the path $P_{u', v_{M_2}}$ and the $4$-hop subtree rooted at $w$ does not contain any marked nodes.
 Hence \Cref{lem:duOne}, \ref{lem:noDieOut} implies that this subtree contributes at least $4$ nodes to $\tree_{t+6}$.

If $u$ has no marked neighbor, the $6$-hop subtree rooted at $u$ may contain a marked node. If the subtree does not contain a marked node, it contributes at least $(\Delta-1)^{3} = 8$ nodes to $\tree_{t+6}$ due to \Cref{lem:duOne}, \ref{lem:noDieOut}. If the subtree rooted at $u$ contains a marked node let $u'$ be the unique neighbor on level $t$ that $u$ must have. We show that the $6$-hop subtree rooted at $u'$ does not contain a marked node: Let $v_{M_1}$ be the marked node in the subtree of $u$ created by $\Tnode$ $v_T$. The subtree rooted at $u'$ cannot contain $v_{M_1}$ as otherwise $P_{u,v_{M_1}}, P_{u',v_{M_1}}, \{u,u'\},P_{v',u}, P_{v',u'}$ induces a small DCC where $v'$ is the least common ancestor of $u$ and $u'$. The subtree also cannot contain a marked node that stems from any $\Tnode$ other than $v_T$ due to $b=15$. Thus let $v_{M_2}$ be the other marked node that $v_T$ created. $v_{M_2}$ cannot be a node in the $6$-hop subtree rooted at $u'$ because then $P_{u,v_{M_1}}, \{v_{M_1},v_T\},\{v_{M_2},v_T\}, P_{u',v_{M_2}}, \{u,u'\},P_{v',u}, P_{v',u'}$ induces a small DCC where $v'$ is the least common ancestor of $u$ and $u'$. Hence the $6$-hop subtree rooted at $u'$ does not contain any marked nodes and contributes $8$ nodes to $\tree_{t+6}(v)$ due to \Cref{lem:duOne}, \ref{lem:noDieOut}. Therefore $u$ and $u'$ together contribute at least $8$ nodes to $\tree_{t+6}(v)$.

\smallskip

\emph{Case $d(u) = 0$.} Node $u$ must have one marked neighbor (on level $t+1$), one parent and one unmarked neighbor $u'$ on level $t$ of the BFS tree. The $6$-hop subtree rooted at $u'$ does not contain a marked node: Assume that it does contain a marked node $v_{M_2}$ and let $v_{M_1}$ be the marked neighbor of $u$. First note that $v_{M_1}\neq v_{M_2}$ and let $v_T$ be the $\Tnode$ that marked both nodes. Let $v'$ be the least common ancestor of $u$ and $u'$. Then $P_{v',u}, P_{v',u'}, \{u,u'\}, \{u,v_{M_1}\}, \{v_{M_1},v_{T}\}, \{v_T,v_{M_2}\},P_{v_{M_2},u'}$ induces a small DCC, a contradiction.

Then \Cref{lem:duOne}, \ref{lem:noDieOut} implies that the $6$-hop subtree rooted at $u'$ contributes at least $(\Delta-1)^{3}=8$ unmarked nodes on level $t+6$, that is, $u$ and $u'$ together contribute at least $8$ nodes.

\medskip

In every case we obtain at least $4$ unmarked reachable nodes on level $\tree_{t+6}$ per node on level $\tree_t$, which proves the claim.
\end{proof}
\subsection{A Simplified Proof for the Distributed Brooks' Theorem} \label{ssec:distributed-brooks}
Panconesi and Srinivasan proved a distributed version of Brooks' Theorem (cf. \Cref{thm:local-brooks}).
The goal of this section is to provide a simplified proof of the result. We begin by observing that if a graph does not have any small degree-choosable components, it is locally expanding. This result is easier to prove than \Cref{lem:detExpansion} as it does not include the marking process.

\begin{lemma} \label{lem:no-dcc-expansion}
	Let $G$ be a graph and $v \in V(G)$ be a node such that inside the $r$-radius neighborhood of $v$ there are no degree-choosable components and every node has degree $\Delta$. Then for each even $r$ there are at least $(\Delta-1)^{r/2}$ nodes at distance $r$ from $v$.
\end{lemma}

\begin{proof}
	Consider the BFS tree around node $v$. The $1$-hop neighborhood of $v$ consists of $\Delta$ nodes that form disjoint cliques due to \Cref{lem:disjointCliques}. As not all neighbors can form a single clique each such neighbor has at least one edge to level two of the BFS tree. This implies that $|B_2(v)| \geq \Delta\geq \Delta-1$.

	Now consider some level $t$. For each node $u \in B_t(v)$, we get via \Cref{lem:deg-sum} that the number of descendants of $u$ in $B_{t+2}(v)$ is at least
	\[
		\sum_{u' \in N_+(u)} d(u') \geq d(u)(\Delta-d(u)) \geq \Delta-1.
	\]
	Since each node has a unique ancestor in the BFS tree, we get that
	\begin{align*}|B_{t+2}(v)| \geq (\Delta-1)|B_t(v)|~.& \qedhere \end{align*}
\end{proof}

Now we can use previous lemmas to show that the uncolored node in the statement of \Cref{thm:local-brooks} can fix its color as it sees a degree-choosable component or a node of degree $< \Delta$ inside its $O(\log_{\Delta} n)$-neighborhood.

\begin{lemma}\label{lem:dcc-or-expansion}
	Let $G$ be a graph with maximum degree $\Delta\geq 3$. The $(2\log_{\Delta-1} n)$-neighborhood of any node contains a degree-choosable component or it contains a node of degree smaller than $\Delta$.
\end{lemma}

\begin{proof}
		Fix a node $v\in V(G)$ and assume that its $r = 2\log_{\Delta-1}n$ neighborhood does not contain a degree-choosable component and that nodes in this neighborhood have degree $\Delta$.
		By \Cref{lem:no-dcc-expansion}, the BFS tree has $|B_r(v)| \geq (\Delta-1)^{r/2} \geq n$ nodes. Therefore the BFS tree cannot expand, and there is an edge to a lower level of $\BFS(v)$ from $B_r(v)$, or there is a node of degree $< \Delta$ in $B_r(v)$.
\end{proof}

Now we are ready for the proof of the distributed Brooks' theorem.

\begin{proof}[Proof of Theorem~\ref{thm:local-brooks}]
Let $c$ denote the partial coloring $G$, with $c(v) = \bot$. We say that $v$ has a \emph{token}. We can always do the following operation: let $u$ be an arbitrary neighbor of $v$. If $v$ does not have a free color, that is, all of its $\Delta$ neighbors have $\Delta$ different colors, then we can move the token to $u$ and color the node $v$ with color $c(u)$. If $v$ has a free color, it can choose that color and the token is eliminated. Now, if the $(2\log_{\Delta-1} n)$-neighborhood of $v$ contains a node of smaller degree, we can move the token to that node, and it is guaranteed to have a free color.
	Now assume that no such node exists. By \Cref{lem:dcc-or-expansion}, there exists a degree-choosable component in the $(2\log_{\Delta-1} n)$-neighborhood of $v$. Let $u$ be one of the closest nodes to $v$ in the degree-choosable component $B$. We move the token from $v$ to $u$ by the shortest path. Next we uncolor the whole $B$. By definition there exists a $\Delta$-coloring of $B$ compatible with the existing coloring in the rest of the graph. \qedhere
\end{proof}

\begin{remark}
\Cref{thm:local-brooks} implies an $\SLOCAL(O(\log_{\Delta}n))$ algorithm (see \cite{SLOCAL17} for the model). This combined with \cite[Theorem 1.11]{SLOCAL17} immediately implies the existence of a randomized polylogarithmic round algorithm for $\Delta$-coloring. Note that \cite{Panconesi1995} explicitly gives such an algorithm and we provide a faster randomized algorithm in \Cref{thm:mainDelta}.
\end{remark}


\section{\texorpdfstring{Algorithmic Preliminaries \& Notation}{Algorithmic Preliminaries \& Notation}}
\label{sec:preliminaries}

Given a subset $C\subseteq V$ of nodes of a graph $G=(V,E)$,  $C$ has \emph{weak diameter} $d$ if $d_{G}(v,w)\leq d$ for all $v,w\in C$. 
\begin{definition}[Network Decomposition, \cite{awerbuch89}]
\label{def:decomposition}
  A weak \emph{$\big(d(n),c(n)\big)$-network-decomposition} of an
  $n$-node graph $G=(V,E)$ is a partition of $V$ into clusters such
  that each cluster has weak  diameter at most $d(n)$ and the
  cluster graph is properly colored with colors $1,\dots,c(n)$.
\end{definition}
One can compute a decomposition with $d(n), c(n)=2^{O(\sqrt{\log n})}$ in $2^{O(\sqrt{\log n})}$ rounds \cite{panconesi1992improved}.

In the $(\deg+1)$ list coloring problem each node $v$ has a list $L(v)$ of available colors with $|L(v)|\geq \deg(v)+1$. The objective is to properly color the graph such that each node picks a color from its list.
\begin{theorem}[\cite{fraigniaud15} +\cite{Elkin17}]
\label{thm:listColoring}
There is a deterministic distributed algorithm that solves the $(\deg+1)$ list coloring problem in time $O\big(\sqrt{\Delta\log\Delta}\cdot\logstar \Delta\big)$ given a $O(\Delta^2)$ coloring of the graph.
\end{theorem}

By iterating through the color classes and greedily picking colors we obtain the following. 
\begin{theorem}[\cite{panconesi1992improved}]
Given a weak \emph{$\big(d(n),c(n)\big)$-network-decomposition} one can solve the $(\deg+1)$ list coloring problem in time $O(c(n)\cdot (d(n)+1))$. 
In particular $(\deg+1)$ list coloring can be solved in $2^{O(\sqrt{\log n})}$ rounds. 
\end{theorem}

\begin{theorem}[List Coloring \cite{ghaffari16improved}]\label{lem:randListColoring}
There is a randomized distributed algorithm that solves the $(\deg+1)$-list coloring problem in $O\big(\log \Delta +2^{O(\sqrt{\log \log n})}\big)$ rounds.
\end{theorem}

For some graph $G$ and an integer $R$ let $G^R=(V(G),\{\{u,v\}\mid \dist_G(u,v)\leq R\})$.
An $(\alpha, \beta)$ ruling set of a graph $G$ is a subset $M\subseteq V(G)$ of the nodes such that $\dist(v, M\setminus \{v\})\geq \alpha$ for all $v\in M$ and $\dist(v,M)\leq \beta$ for all $v\in V$. If $\alpha=1$, we also omit the parameter speak of \emph{$\beta$-ruling sets}.
Usually, when computing a ruling set it comes with a so called \emph{ruling forest}. 
\begin{definition}[Ruling Forest, \cite{awerbuch89}]
Given a graph $G = (V,E)$ and a set $V'\subseteq V$, we say that a spanning forest $F_r = (V_r, E_r)$, $V’ \subseteq V_r$ is an \emph{$(\alpha, \beta)$-ruling forest} with respect to $V'$ if the following three conditions holds:
\begin{enumerate}[label=(\arabic*)]
\item The root of the trees are in $V'$,
\item the distance in $G$ between any two roots is at least $\alpha$,
\item the depth of each tree in the forest is at most $\beta$.
\end{enumerate}
\end{definition}
The following lemma  summarizes the known distributed ruling set algorithms that we use.
\begin{lemma}
\label{lem:rulingSets}
For any integers $k,\beta\geq 1$ there are the following ruling set algorithms.

 \smallskip

		\begin{tabular}{cllll}
			(1) & $(2,\beta)$ & Det. &\cite{schneider13} \label{item:schneiderrulingDet} & $O(\beta\Delta^{2/\beta}+\logstar n)$ \\
			(2) & $(k,k^2\beta)$ & Det. & \cite{schneider13} + \cite{BEPSv3}& $O(k^2\cdot\beta\Delta^{2/\beta}+k\cdot\logstar n)$    \label{item:rulingDet}\\
			(3) & $(2, O(\log \log n))$ & Rand. &  \cite{Gfeller07} +\cite{schneider13} \label{item:gfeller} & $O(\log \log n)$ \\
			(4) & $(2,\beta)$ &  Rand. & \cite{ghaffari16improved} \label{item:ghaffari} & $O(\log^{1/\beta}\Delta)+2^{O(\sqrt{\log\log n})}$ \\
			(5) & $(2,1)$ &  Det. & \cite{awerbuch89,panconesi1992improved}  & $2^{O(\sqrt{\log n})}$ \\
		\end{tabular}
\end{lemma}

\begin{proof}[Proof of \Cref{lem:rulingSets}]
We only provide a proof for the deterministic $(k,k^2\beta)$ algorithm which is alike the words in \cite[Section 1.1]{BEPSv3}. Consider $G^{k-1}$  and note that the maximum degree of $G^{k-1}$ can be upper bounded by $\Delta^k$.
Then apply \Cref{lem:rulingSets} (1) on $G^{k-1}$ with $\beta'=k\beta$. Note that each step of the algorithm can be executed in $k$ steps in $G$ leading to a runtime of 
\begin{align*}O(k\cdot (\beta'\Delta^{2k/\beta'} +\logstar n))=O(k^2\cdot\beta\Delta^{2/\beta}+k\cdot\logstar n)~.& \qedhere\end{align*}
\end{proof}

\section{\texorpdfstring{Deterministic $\Delta$-coloring (Theorem \ref{thm:deterministicDelta})}{Deterministic Delta-coloring (Theorem \ref{thm:deterministicDelta})}}
\label{sec:detColoring}
In this section we present our deterministic $\Delta$-coloring algorithm, exemplifying our layering technique.
\begin{figure}[ht]
\centering
\includegraphics[width=0.7\textwidth]{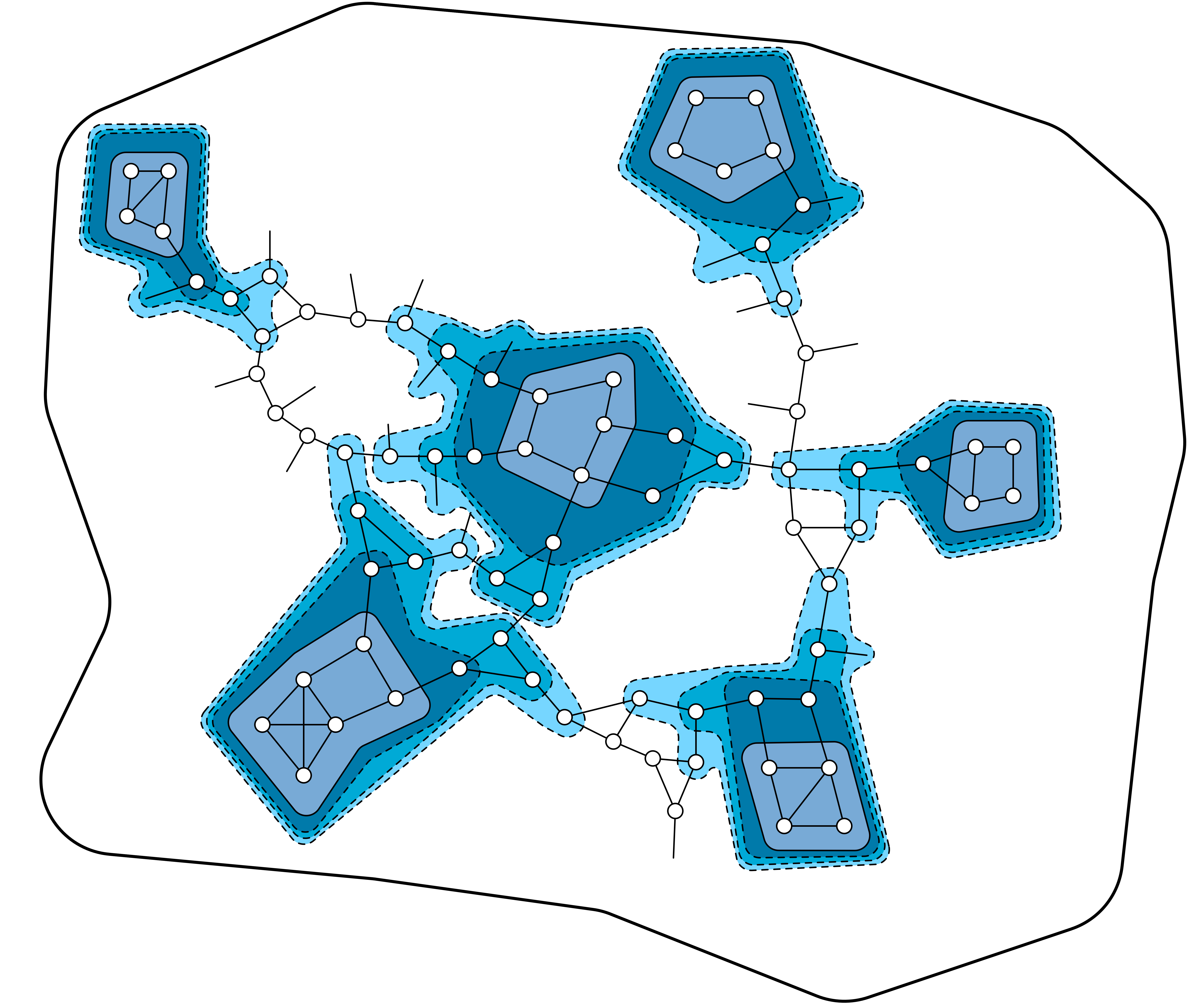}
\caption{The layering technique: First, the base layer is removed, then neighbors of the base layer are removed, then their neighbors and so on. Attention! The base layer in the illustrated example is build as in our main algorithm in \Cref{sec:delta-alg} and consists of degree choosable components. To illustrate the technique in the setting of \Cref{sec:detColoring} simply assume that each component of the base layer consists of a single node. This illustration only shows one connected component of the base layer and in general the illustrated layer building starts at several places of the graph (wherever nodes are in the base layer) at the same time.}
\label{figure:deltaColoring}
\end{figure}

\paragraph{Layering Technique.} In the layering technique there is a carefully chosen base layer $B_0$ that is easy to color and layers $B_1,\ldots,B_s$ where $B_i$ consists of the nodes in distance $i$ to $B_0$. This is particularly helpful for $\Delta$-coloring as we can $\Delta$-color the layers in \underline{reverse order} while respecting the colored neighbors in layers with a larger index. To $\Delta$-color layer $B_i, i\neq 0$ we need to solve a $(\deg+1)$ list coloring on the graph $G[B_i]$: A node $v\in B_i$ builds its list by removing the colors of neighbors in $B_{i+1}\cup\ldots \cup B_s$ from the set $\{1,\ldots,\Delta\}$. The size of this list is at least $\deg_{G[B_i]}(v)+1$ as $v$ has one neighbor in layer $B_{i-1}$. Then layer $B_0$ is colored after all other layers with different techniques as $\Delta$-coloring $B_0$ while respecting already colored neighbors might not be a $\deg+1$ list coloring instance. To make sure that we can still $\Delta$-color $B_0$ efficiently (we might have to recolor previously colored nodes) it has to be chosen carefully. Also consult \Cref{figure:deltaColoring} for an illustration of the technique. 

The deterministic algorithm in this section uses the layering technique in a simple setting with  $O(\log^2 n)$ layers. The layer $B_0$ is chosen to be a ruling set in which nodes have large distances such that \Cref{thm:local-brooks} can be applied to color the nodes in $B_0$ independently and in parallel. 
The total runtime is dominated by the $O(\log^2 n)$ iterations of list coloring due to the layering technique.


\paragraph{Algorithm.}
 First, color all nodes of $G$ with $O(\Delta^2)$ colors with Linial's algorithm \cite{linial92}. These colors are only used for symmetry breaking when applying list coloring algorithms and do not coincide with the desired $\Delta$-coloring. Let $R:=4\log_{\Delta-1} n +1\leq 7\log n /\log \Delta$ and $z=4\cdot R^2$.
\begin{enumerate}[label=(\arabic*)]
\item  \myit{Build layer $B_0$} Compute a $(R,z)$ ruling forest of $G$ with \Cref{lem:rulingSets}, (2). Add all nodes of the ruling set to layer $B_0\subseteq V$.
\item  \myit{Remove layers  $B_0,\ldots ,B_{z}$} Define layers  $B_1,\ldots ,B_{z}$ where $v\in B_i$ if the distance of $v$ to $B_0$ is $i$. Remove all layers from the graph.

\item \myit{Color layers $B_z, \ldots, B_1$ } Add the layers $B_z, \ldots, B_1$ to the graph one by one: When adding layer $B_i$ color the nodes of $B_i$  such that $\GraphI=G[\bigcup_{j=i}^{z} B_j]$ is validly $\Delta$-colored. Step $i=z,\ldots,1$ is a $\deg+1$ list coloring instance on $G_i=G[B_i]$ because a node $v\in B_i$ has an uncolored neighbor in $B_{i-1}$~.
We use  \Cref{thm:listColoring} to solve each list coloring instance.

\item \myit{Color layer $B_0$}  Use \Cref{thm:local-brooks} to independently color the nodes in $B_0$ through recoloring nodes within distance at most $2\log_{\Delta-1} n<R/2$.
\end{enumerate}

\begin{proof}[Proof of \Cref{thm:deterministicDelta}]
By the definition of a ruling set every node of $G$ is in distance at most $z$ from its root in the ruling forest. Thus every node is contained in  the $z+1$ layers and  is colored.

We formally show that coloring each layer is an instance of $\deg+1$ list coloring in the graph $G_i$. Assume that we are in step $i$ and want to color the nodes of $B_i$ such that $\GraphI$ is validly $\Delta$-colored. Pick a node $v\in B_i$. The list of available colors of $v$  is $\{1,\ldots,\Delta\}\setminus F_v$ where $F_v$ is the set of colors that have already been chosen by $v$'s colored neighbors in $\GraphI$.  The size of $F_v$ is at most $\deg_{\GraphI}(v)-\deg_{G_i}(v)$.
The degree $\deg_{\GraphI}(v)$ is upper bounded by $\Delta-1$ because at least one of $v$'s neighbors in $G$ is contained in $B_{i-1}$. Thus  the list of available colors of $v$ has size at least $\Delta-|F_v|\geq \Delta-(\deg_{\GraphI}(v)-\deg_{G_i}(v))\geq \deg_{G_i}(v)+1$~.

The runtime of the first step and the second step is $O\big(R^2\cdot\sqrt{\Delta}+R\cdot \logstar n+z+\logstar n\big)=O\big(R^2\cdot\sqrt{\Delta}\big)$. The third step takes $O(\sqrt{\Delta\log \Delta}\logstar \Delta)$ rounds for each of the $z=O(R^2)=O(\log_{\Delta}^2 n)$ iterations. The fourth step takes $O(R)$ rounds. In total the runtime is dominated by the third step.
\end{proof}

The following theorem appeared as \cite[Theorem 5]{Panconesi1995}. Our techniques can be used to give an alternative proof.
\begin{theorem}[\cite{Panconesi1995}, rephrased, reproved]
\label{thm:deterministicNetComp} Nice graphs can be $\Delta$-colored deterministically in the distributed model
  of computation in $2^{O(\sqrt{\log n})}$ rounds.
\end{theorem}

\paragraph{Algorithm.}
\begin{enumerate}[label=(\arabic*)]
\item \myit{Build layer $B_0$} Set $R=4\log_{\Delta-1} n +1\leq 7\log n /\log \Delta$ and compute a $\big(2^{O(\sqrt{\log n})}, 2^{O(\sqrt{\log n})}\big)$ network decomposition of $G^R$ with \cite{panconesi1992improved}. Then compute an $(R, R+1)$ ruling set with the help of the network decomposition in  $2^{O(\sqrt{\log n})}$ rounds.  Let $B_0$ be the nodes in the ruling set.

\item \myit{Remove layers  $B_0,\ldots ,B_{z}$} Define layers  $B_1,\ldots ,B_{z}$ where $z=R+1$ and $v\in B_i$ if the distance of $v$ to $B_0$ is $i$. Remove all layers from the graph.

\item \myit{Color layers $B_z, \ldots, B_1$} Add the layers $B_z, \ldots, B_1$ to the graph one by one: When adding layer $B_i$ color the nodes of $B_i$  such that $\GraphI=G[\bigcup_{j=i}^{z} B_j]$ is validly $\Delta$-colored. Step $i=z,\ldots,1$ is a $\deg+1$ list coloring instance on $G_i=G[B_i]$ because a node $v\in B_i$ has an uncolored neighbor in $B_{i-1}$.
Use the network decomposition to solve the list colorings (note that a network decomposition of $G^R$ is also a network decomposition of $G$ with an $R$ factor increase in the diameter of the clusters).
\item  \myit{Color $B_0$} Use \Cref{thm:local-brooks} to independently color the nodes in $B_0$ through recoloring nodes in distance at most $2\log_{\Delta-1} n<R/2$.
\end{enumerate}

\begin{proof}
The proof of correctness is along similar lines as the proof of \Cref{thm:deterministicDelta}.
The runtime of the third step dominates and is $O\big(z\cdot 2^{O(\sqrt{\log n})}\big)=2^{O(\sqrt{\log n})}$.
\end{proof}


\section{\texorpdfstring{Randomized $\Delta$-Coloring (Theorem \ref{thm:mainDelta} and \ref{thm:main2})}{Randomized Delta-Coloring ((Theorem \ref{thm:mainDelta} and \ref{thm:main2}))}} \label{sec:delta-alg}

The randomized algorithm is split into two slightly different versions based on $\Delta$: one version can handle any $\Delta \geq 4$ and the other any $3\leq \Delta = O(1)$. We refer to these two versions as the \emph{large-$\Delta$} version and the \emph{small-$\Delta$} version. In this section we present the algorithms  of \Cref{thm:mainDelta,thm:main2} and their proofs. We recommend to read the algorithms in a top-down manner beginning with the captions of the respective parts.

Both variants share the same basic structure. We decompose the graph into layers $B_0, \dots, B_{s},$ $C_0, \dots, C_{2r}$ (and in some cases also layers $D_0, \dots, D_{\alpha}$) of nodes such that all nodes are either colored or are in one of the layers. Then, the layers are iteratively colored in the reverse order that they were built. Coloring a single layer requires solving a $(\deg+1)$-list coloring instance since we will guarantee that each node has an uncolored neighbor in a lower layer.

In Phase~\ref{phase:remove-dccs} we build layers $B_0, \dots, B_{s}$: we identify the dense parts of the graph -- the parts which are easy to color after the rest of the graph has been colored.\footnote{\emph{Dense} in the sense that the part has a small DCCs that by definition can be colored easily after everything else is colored. The term \emph{dense} is inspired as the expansion results in \Cref{sec:structural} show that parts without DCCs are in some sense \emph{sparse}.} These are removed from the graph, along with the nodes around them, to be colored later. Let $H$ denote the remaining graph.

 In Phase~\ref{phase:shattering} we extract layers $C_0, \dots, C_{2r}$ from $H$:
 Phase~\ref{phase:remove-dccs} guarantees that $H$ does not contain any dense parts, and therefore the remaining graph must expand.
We take advantage of this by randomly inserting \emph{slack} into the graph. This means that we pick a set of well-separated nodes and color two of their neighbors with the same color: these nodes now effectively have decreased their degree and will be easy to color later. We again remove the nodes with slack along with the nodes around them to be colored later. Due to the expansion of $H$ we can prove that the probability of each node to remain after this process is small.

Actually, in the small-$\Delta$ case we prove by union bound that with high probability no node remains after this process.  In the large-$\Delta$ case we show that the graph formed by the remaining nodes has \emph{shattered}: remaining connected components are of small size and can be colored efficiently with a similar layering technique using layers $D_0, \dots, D_{\alpha}$ (cf. Phase \ref{phase:smallComponents} in \Cref{sec:highLevel}).

In Phase~\ref{phase:slack-coloring} we color the layers $C_0, \dots, C_{2r}$ in reverse order.
Then, in Phase~\ref{phase:color-dccs}, we color the layers $B_1, \dots, B_{s}$ in the reverse order. By definition $B_0$ consists of (dense) parts that are easy to color if the remaining graph is colored, actually $B_0$ consists of independent DCCs and we can color the components of $B_0$ at the very end.


\subsection{\texorpdfstring{The Randomized $\Delta$-Coloring Algorithms}{The Randomized Delta-Coloring Algorithms}}
\label{sec:highLevel}


First, we remove all degree-choosable components of radius $r$ or less from the graph. This implies that the graph must expand locally (\Cref{lem:dcc-or-expansion}). The two versions differ in the radius $r$: in the small version we choose $r = O(\log \log n)$ and in the large version $r = O(1)$. Let $b=6$ in the large version and $b=15$ in the small version. Set $p=\Delta^{-b}$.

First, color all nodes of $G$ with $O(\Delta^2)$ colors with Linial's algorithm \cite{linial92}. These colors are only used for symmetry breaking when applying list coloring algorithms and do in no way coincide with the desired $\Delta$-coloring.

\clearpage

	\begin{enumerate}[label=\textbf{\Roman*}, leftmargin=0.5cm]
	\item  \mytitle{Removing Degree Choosable Components and Layers around them} \label{phase:remove-dccs}
	\begin{enumerate}[leftmargin=\algorithmInnerMargin]
		\item 
		Each node that is contained in at least one degree choosable subgraph with radius at most $r$ selects one such subgraph. Let $\GraphR$ be the virtual graph that has a node for each selected degree choosable subgraph, and two subgraphs in $V(\GraphR)$ are connected in $\GraphR$ by an edge if they share a vertex or if they are connected by an edge in $G$. The graph $\GraphR$ has at most $n$ nodes as every node adds at most one degree choosable component, maximum degree at most $\Delta^{2r+1}\leq \Delta^{3r}$, and one round of a distributed algorithm in it can be simulated in $O(r)$ rounds in $G$.

\runtimel $O(r)=O(1)$.

		\runtimes $O(r)=O(\log \log n)$.

		\item \myit{Build layer $B_0$}	Find a $(2,\beta)$ ruling set $M$ of $\GraphR$ with $\beta=6\cdot r$. Add each node $v\in V(G)$ that is contained in a DCC  $C\in M$ to the base layer $B_0$.

 		\runtimel With \Cref{lem:rulingSets}, (4) $O\big(\log^{1/\beta} \Delta +2^{O(\sqrt{\log\log n})}\big)$.


    \runtimes With \Cref{lem:rulingSets}, (3) $O(r\cdot(O(\log \log n))) = O(\log^2 \log n)$.

		\label{phase:setR}

		\item \label{phase:grow} \myit{Remove layers $B_1,\ldots, B_s$}
		For $s =\beta\cdot (r+1)$ define layers $B_1,\ldots, B_s$.
		Layer $B_i$ consists of the nodes of $G$ that have distance $i$ (measured in $G$) from a node in $B_0$.
		Remove all layers $B_0,\ldots,B_s$ from the graph.

		\runtimel $O(s)=O(\beta\cdot r)=O(1)$.

		\runtimes $O(s)=O(\beta\cdot r)=O(\log^2 \log n)$.

		\end{enumerate}
		Note that (besides potentially some other nodes)  all nodes that are in a degree choosable component with radius at most $r$ are removed from the graph after phase \ref{phase:grow}.

		\item	\mytitle{Shattering of the Remaining Graph} \label{phase:shattering}
		\begin{enumerate}[resume,leftmargin=\algorithmInnerMargin]
		\item \myit{Random \Tnode creation}
		Consider the remaining graph \[H=G\setminus \big(\bigcup_{i=0}^sB_i\big)~.\]

		\label{phase:marking}
    Each node of $H$ becomes \emph{selected} independently with probability $p$. Then, if there is another selected node within distance $b$, both become unselected. If not, the selected node picks a random pair of non-adjacent neigbors and colors them with color \emph{one}. We call these neighbors \emph{marked}.

		\runtime $O(1)$.
		\item \label{phase:assigningHappyNodes}
		\myit{Remove layers $C_{0},\ldots,C_{2r}$}
			We call a node \emph{happy} if it has an uncolored path to a \Tnode in its $r$-neighborhood. By this definition we assign each happy node to its closest \Tnode in its $r$-neighborhood.

		We define the \emph{boundary} of graph $H$ as the set of nodes with degree less than $\Delta$ in $H$. Nodes  that are colored and have distance at most $r$ steps away from the boundary now remove their color and each such node is assigned to its closest boundary node, breaking ties using identifiers. It might happen that a node $v$ that is $r \leq \ell < 2r$ steps away from the boundary was assigned to a node $w$ that is at most $r-1$ steps away from the boundary. Due to the uncoloring $w$ might not be a \Tnode anymore. However, $w$ is assigned to a node $w'$ on the boundary. Then there is an uncolored path of length at most $2r$ from $v$ to $w'$ through $w$ and we assign $v$ to $w'$ as well.

	Define layers $C_0,\ldots, C_{2r}$, where $C_i$ consists of nodes of $H$ that are at distance $i$ from their respective assigned node. The layer $C_0$ consists of \Tnodes, and all nodes that have degree $<\Delta$ in $H$. Remove the layers $C_0,\ldots, C_{2r}$ and the marked nodes from the graph.

    In \Cref{ssec:smallDelta} we show that in the algorithm for small $\Delta$, all nodes are removed after this phase with high probability. Hence this algorithm proceeds directly to Phase~\ref{phase:colorHappy}.

		\runtimel $O(r)=O(1)$.

    \runtimes $O(r)=O(\log \log n)$.

		\item \myit{Color Small Components} \label{phase:smallComponents}
		Consider the remaining graph $L=H\setminus (C_0\cup \ldots\cup C_{2r}\cup C')$ where $C'$ are the marked nodes. In \Cref{ssec:shattering} we show that the probability for a node of $H$ to remain in $L$ is small and then the standard shattering technique (cf., e.g., \cite{BEPSv3} or \Cref{lem:shattering}) implies that $L$ consists of small connected components of size at most $N:=\poly \Delta\cdot \log_{\Delta}n$~.

		\Cref{sec:smallComponents} explains in detail how to color these small components if $\Delta\geq 4$. The core idea is that we can again handle the small components by constructing layers $D_0,\ldots,D_{\alpha}$ where $\alpha=O(\log^2\log n)$.   Besides some other nodes layer $D_0$ contains the nodes that have an uncolored neighbor in the layers $C_0\cup C_1\cup \ldots\cup C_{2r}$, i.e., they just did not get removed because the closest \Tnode was a little bit too far away. One can show that layer $D_0$ contains at least one node of each small component. However, then all nodes of the component are in one of the layers, because, assuming that a node $v$ of a small component does not see a node of the first few layers, the BFS tree of $v$ within the component expands so fast (basically due to \Cref{lem:detExpansion}) that it sees the whole component in $O(\log_{\Delta} N)=O(\log\log n)$ hops, a contradiction.
		
		\Cref{ssec:smallDelta} shows that for $\Delta=3$ this step can be omitted and $L$ does not contain any vertices.

		\runtimel $2^{O(\sqrt{\log \log n})}$ via \Cref{lem:smallComponents}.

		\runtimes $O\left(\sqrt{\Delta\log \Delta}\logstar\Delta\cdot\log^2\log n\right)$ via \Cref{lem:smallComponents}
		\end{enumerate}

	\item \mytitle{Color Happy Nodes From the Shattering Process} \label{phase:slack-coloring}
   		\begin{enumerate}[resume,leftmargin=\algorithmInnerMargin]
		\item \label{phase:colorHappy}
\myit{Color layers $C_{2r},\ldots,C_0$}
		Assume that the remaining small components are colored with $\Delta$ colors in Phase~\ref{phase:smallComponents}. Go through the layers $C_{2r},\ldots,C_0$ grown in step~\ref{phase:assigningHappyNodes} in reverse order and $\Delta$-color them one at a time while respecting the colors of nodes that are already colored. Coloring layer $C_i$ corresponds to a $(\deg+1)$-list coloring instance on $H[C_i]$, since for each $i=2r,\ldots, 1$ each node has a neighbor at a lower level and the nodes in $C_0$ have two neighbors of the same color. 

		\runtimel $O(\log \Delta + 2^{O(\sqrt{\log \log n})})$ with using \Cref{lem:randListColoring} $2r=O(1)$ times.


    \runtimes $O(\log \log n (\sqrt{\Delta\log \Delta}\cdot\logstar \Delta))$ with using \Cref{thm:listColoring} $2r$ times.
		\end{enumerate}
		\item \mytitle{Color Degree Choosable Components and Layers around them} \label{phase:color-dccs}
		\begin{enumerate}[resume,leftmargin=\algorithmInnerMargin]

		\item \label{phase:colorLayers}
\myit{Color layers $B_s,\ldots,B_1$}
		Go through the layers $B_s,\ldots,B_1$ grown in step~\ref{phase:grow} and color each layer with $\Delta$ colors while respecting nodes colored previously: Coloring layer $B_i$ forms a $(\deg'+1)$-list coloring instance on $G[B_i]$, since each node has an uncolored neighbor in $B_{i-1}$. 

		\runtimel $O(\log \Delta + 2^{O(\sqrt{\log \log n})})$ with using \Cref{lem:randListColoring} $s=O(1)$ times.


      \runtimes $O(\log \log n (\sqrt{\Delta\log \Delta}\cdot\logstar \Delta))$ with using \Cref{thm:listColoring} $s$ times.

		\item \myit{Color layer $B_0$} Each connected component of layer $B_0$ corresponds to one DCC in $M$ (selected in step~\ref{phase:setR}). Thus  each  connected component of $B_0$ is $\Delta$-list colorable and of  radius $\leq r$. We find a coloring by brute forcing each component independently.

		\runtimel $O(r)=O(1)$.

    \runtimes $O(r)=O(\log \log n)$.
	\end{enumerate}
\end{enumerate}
\begin{figure}
\begin{subfigure}{.5\textwidth}
  \centering
  \includegraphics[width=0.8\textwidth]{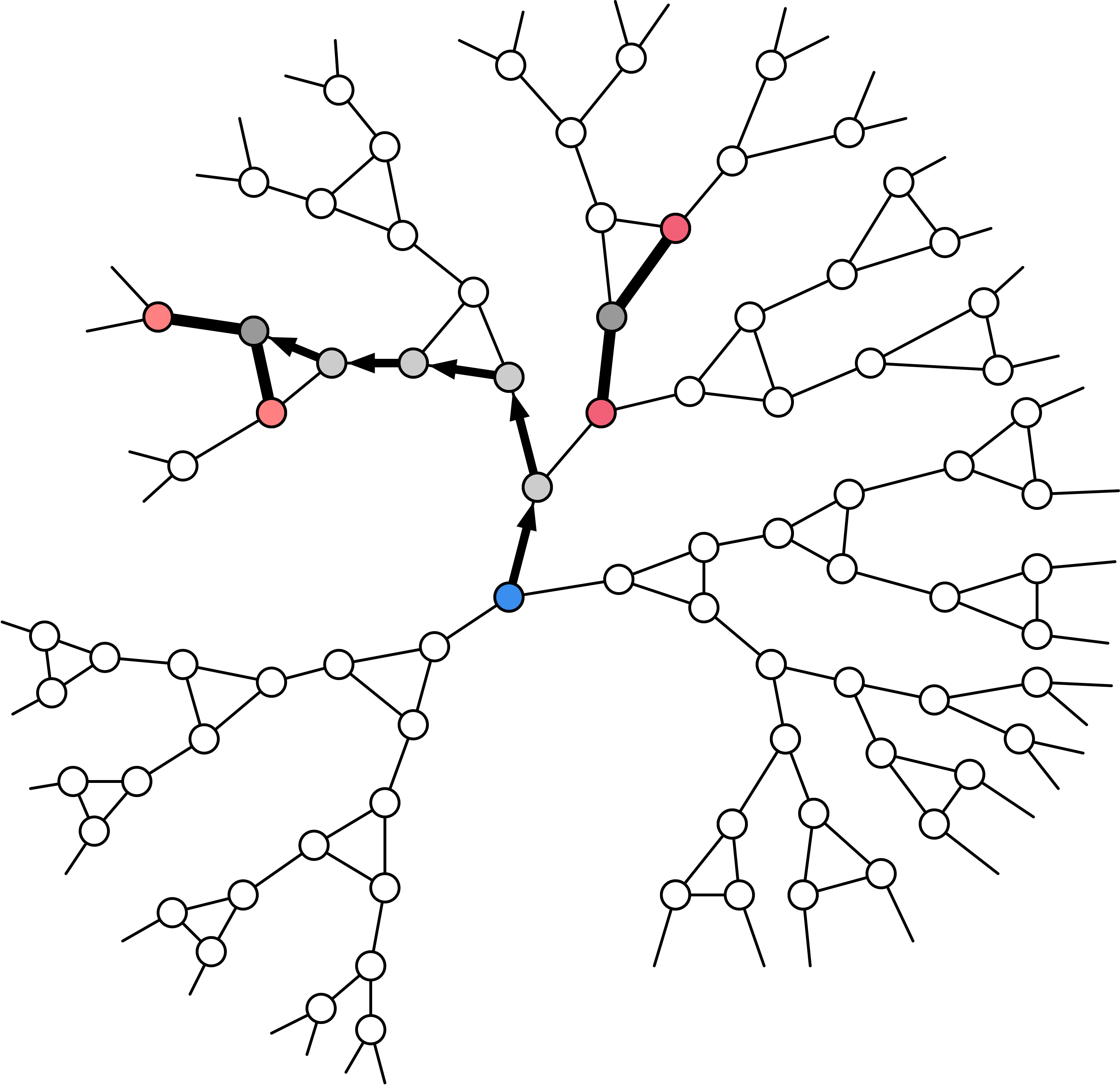}
  \label{figure:slackCreation1}
\end{subfigure}%
\begin{subfigure}{.5\textwidth}
  \centering
  \includegraphics[width=0.8\textwidth]{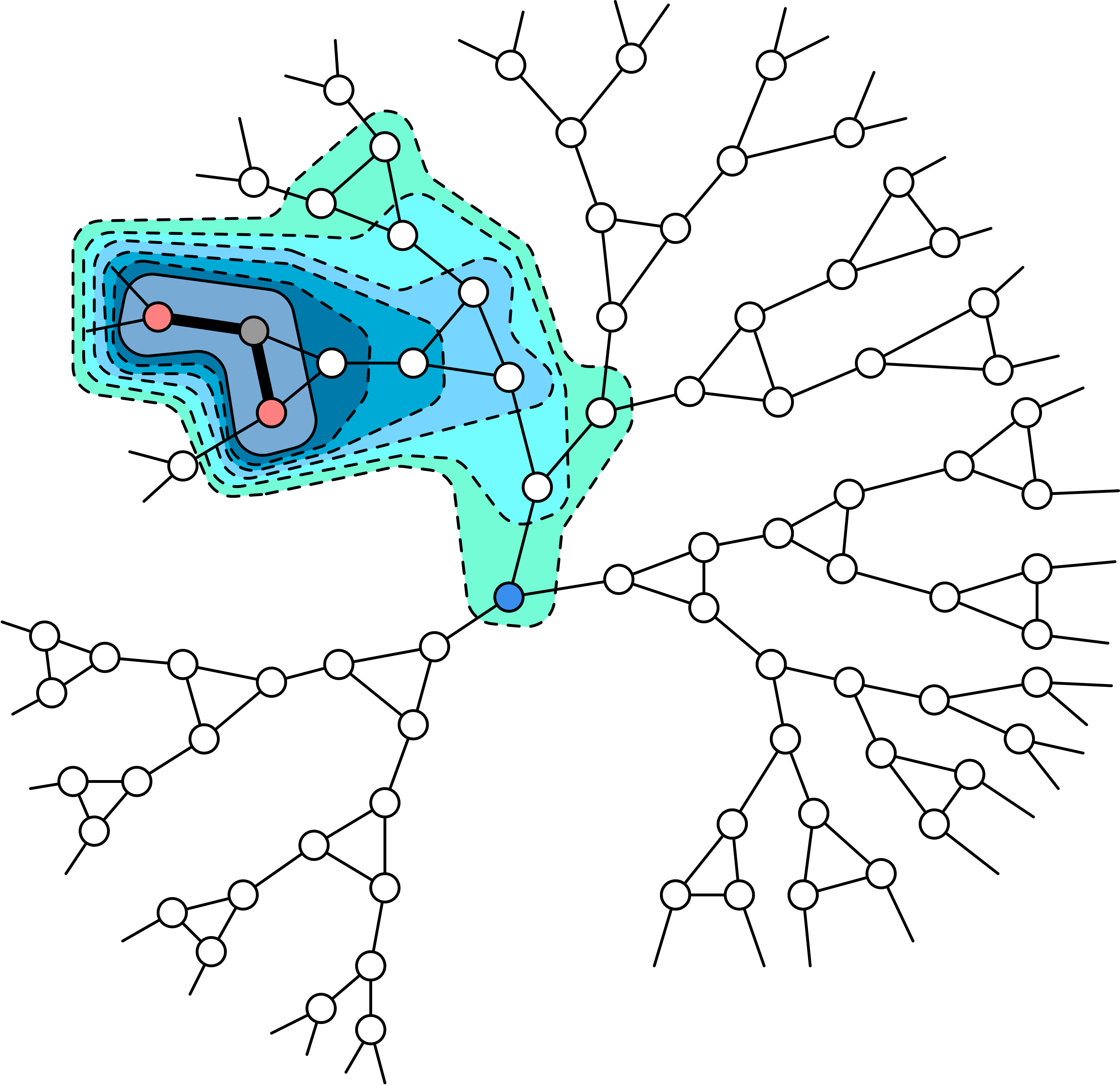}
 \label{figure:slackCreation2}
\end{subfigure}
\caption{The creation of \Tnodes in Phase \ref{phase:marking}. Note that the top right \Tnode is not reachable from the blue node in the center through a path of unmarked nodes. To be removed in the layer creation around the \Tnodes it is essential that a \Tnode is reachable through an unmarked path. The image is also helpful to get an understanding of how the locally Gallai tree like graph $H$ looks like (a tree of cliques and odd cycles). However, note that certain edges are left out to simplify the illustration.
In the right hand side illustration you can see how the \Tnode is removed from the layer creation around the \Tnode because it has a short path of unmarked nodes to the \Tnode node. 
}
\end{figure}

\begin{proof}[Proof of \Cref{thm:mainDelta}]
The runtime follows from summing up the runtimes for small $\Delta$ of all phases and is dominated by the runtime of phases (7) and (8) in which we need to solve $O(r+s)=O(\log^2\log n)$ list coloring instances in $O(\sqrt{\Delta\log\Delta}\logstar\Delta)$ rounds each. Solving the small components in phase (6) also has a significant contribution and can be done in $O(\sqrt{\Delta\log\Delta}\logstar\Delta\cdot \log^2\log n)$ rounds via \Cref{lem:smallComponents} if $\Delta\geq 4$. If $\Delta=3$ then \Cref{ssec:smallDelta} shows that phase (6) can be omitted as $L$ is the empty graph. The ruling set in phase (2) can be found in $O(\log^2\log n)$ rounds. All other phases take at most $O(r+s)=O(\log^2\log n)$ rounds. 

All nodes are colored at the end because any nodes that is in none of the layers \[B_0,\ldots,B_S, C_0,\ldots,C_{2r}, C'\]  is contained in a 'small component' and colored in phase (6), w.h.p. In \Cref{ssec:shattering} and \Cref{sec:smallComponents} we prove that this is indeed the case for $\Delta\geq 4$. For $\Delta=3$ we show in \Cref{ssec:smallDelta} that the aforementioned layers already contain all vertices of the graph w.h.p, i.e., w.h.p. the graph $L$ defined in phase (6) does not have any vertices. 
\end{proof}

\begin{proof}[Proof of \Cref{thm:main2}]
The runtime follows from summing up the runtimes for large $\Delta$ of all phases and is dominated by the runtime of phases (7) and (8) in which we need to solve $O(r+s)=O(1)$ list coloring instances in $O(\log \Delta) + 2^{O(\sqrt{\log\log n})}$ rounds. Solving the small components in phase (6) also has a significant contribution and can be done in $2^{O(\sqrt{\log\log n})}$ rounds via \Cref{lem:smallComponents}. The ruling set in phase (2) can be found in $O(\log \Delta) + 2^{O(\sqrt{\log\log n})}$ rounds. All other phases take $O(r)=O(1)$ rounds. 

All nodes are colored at the end because any nodes that is in none of the layers \[B_0,\ldots,B_S, C_0,\ldots,C_{2r},C'\]  is contained in a 'small component' and colored in phase (6), w.h.p. In \Cref{ssec:shattering} and \Cref{sec:smallComponents} we prove that this is indeed the case.   
\end{proof}

\subsection{Shattering of the Remaining Graph (Phases \ref{phase:marking}-\ref{phase:smallComponents})} \label{ssec:shattering}
In this section we show that the process of phase~\ref{phase:marking} and \ref{phase:assigningHappyNodes}  produces a graph with remaining components of size $O(\poly(\Delta)\cdot\log n)$. In \Cref{sec:smallComponents} we show how to color these small components fast. The nodes that are put into the layers in phase \ref{phase:assigningHappyNodes} are colored later in phase \ref{phase:colorHappy}.

For a node $v$ let $\mathcal{E}_v$ be the event that $v$ is removed in the graph in phases \ref{phase:marking}-\ref{phase:assigningHappyNodes}. Let  $t$ be the radius such that the event $\mathcal{E}_v$ only depends on the random bits of nodes in radius $t$ of $v$.
 The standard shattering technique (cf. \Cref{lem:shattering}) shows that the connected components of non-removed nodes are \emph{small} if the probability of $\overline{\mathcal{E}_v}$ is upper bounded by $1/\poly(\Delta)$ where the polynomial depends on the radius $t$.

To show that the probability of $\overline{\mathcal{E}_v}$ is small enough we show that
 the BFS tree of uncolored nodes around $v$ expands exponentially. Thus after $O(1)$ steps of expansion we see  uncolored paths to enough nodes that independently form a \Tnode with probability $\Theta(p)$ and the probability that none of them actually is a \Tnode will be at most $1/\poly(\Delta)$ for a sufficiently small polynomial.

\label{sec:shatterinProbability}
Now, we upper bound the probability that a given node does not become happy after the shattering process. Due to \Cref{lem:detExpansion} the BFS tree around a node expands deterministically even after the marking process which implies the next lemma.
\begin{lemma}\label{lem:shatteringSet}
For every $0<t\leq r$ and after the selection and marking process the $t$-neighborhood of every node $v$ contains a boundary node  or a set of nodes $S_v$ with the following properties:
\begin{enumerate}[label=(\arabic*)]
\item $|S_v|\geq (\Delta-2)^{t/2}\cdot\Delta^{-6}$~,
\item all nodes in $S_v$ are reachable through uncolored nodes from $v$~,
\item for each $u\in S_v$ the probability that it is selected and creates a \Tnode that does not block the path to $v$ is at least $1/3\cdot p(1-p)^{\Delta^6}$. The events are independent for distinct $u\in S_v$~,
\item For each $u\in S_v$ the event that it forms a \Tnode of the above type only depends on the random bits of nodes in radius $t+7$ around $v$.
\end{enumerate}
\end{lemma}
\begin{proof}[Proof of \Cref{lem:shatteringSet}]
For a fixed node and due to \Cref{lem:detExpansion} the BFS tree around $v$ restricted to unmarked nodes contains at least $(\Delta-2)^{t/2}$ nodes on level $t$ or we encounter a $\Tnode$. Let $A_v$ be the set of these nodes. For each node $u\in A_v$ whose children in the BFS tree form a $\Delta-1$ clique we remove $u$ from $A_v$ and add one of its children $u'$ in the BFS tree to $A_v$. As the child has the $\Delta-2$ nodes of the clique on its own level and $u$ as parent it has only one child in the BFS tree. Thus the children of $u'$ in the BFS tree cannot form a $\Delta-1$ clique. Furthermore, $u'$ is distinct from all other nodes in $A_v$ as the BFS tree is unique.

Now, we greedily add nodes of $A_v$ to $S_v$. When we add a node $u\in A_v$ to $S_v$ we remove the nodes from $A_v$ that are in the $6$-neighborhood of $u$; these are at most $\Delta^6$ many.
Thus the size of $S_v$ is at least $|A_v|\cdot \Delta^{-6}=(\Delta-2)^{t/2}\cdot\Delta^{-6}$ and nodes in $S_v$ have pairwise distance at least $7$.

We now compute the probability that a node $u\in S_v$ is selected and creates a \Tnode that does not block the path to $v$. To ensure that the path to $v$ is not blocked we  (1) condition on the event that certain nodes in the BFS tree around $v$ are not uncolored (through the usage of \Cref{lem:detExpansion}) and (2) we ensure that none of the two nodes that $u$ colors is the single neighbor $u'$ of $u$ that lies on the unique path in the BFS tree to $v$.
Node $u$ is selected with probability $p$ and stays selected if no neighbor in its $6$-neighborhood is selected, i.e., at least with probability $(1-p)^{\Delta^6}$. As $u$ does not have a $\Delta-1$ clique on the next level of the BFS tree there are at least two non adjacent neighbors $u_1$ and $u_2$ of $u$ that are distinct from $u'$. So the probability that $u$ does not mark $u'$ is at least $1/3$.

In this whole process we expanded for $t$ steps to obtain the set $A_v$. The set $S_v$ contains nodes in distance at most $t+1$ from $v$ and we use that nodes in distance $6$ to nodes in $S_v$ are not selected, i.e., the probabilities only depend on the $t+7$ radius of $v$. The event whether  distinct nodes in $S_v$ can generate $\Tnodes$ are independent as they have pairwise distance at least $7$.
\end{proof}

For any node $v$  \Cref{lem:shatteringSet} provides a large set of independent nodes that have uncolored paths to $v$. Thus we can upper bound the probability that a node remains after the shattering process.
\begin{lemma}[Shattering Probability]
\label{lem:shatteringProb}
Let $\Delta\geq 4$. There is an $r=O(1)$ such that every node finds an uncolored path of length at most $r-7$ to a \Tnode with probability at least $1-\left(\frac{1}{\Delta}\right)^{4r+4}$ using only the randomness in its $t$ neighborhood. The constant $r$ is independent from the graph (including its size).
\end{lemma}
\begin{proof}
Let $v$ be a node in $H$. Apply \Cref{lem:shatteringSet} with $t=r-7$ and obtain a set $S_v$ in which each node independently forms a \Tnode that is reachable from $v$ through an uncolored path with probability $1/3\cdot p(1-p)^{\Delta^6}$. The probability that $v$ remains after phase \ref{phase:assigningHappyNodes} is upper bounded by
\begin{align*}
\left(1-\frac{1}{3}p(1-p)^{\Delta^6}\right)^{|S_v|}\leq e^{-\frac{|S_v|}{3}p(1-p)^{\Delta^6}}=e^{-\frac{(\Delta-2)^{t/2}\cdot\Delta^{-12}}{12}}\stackrel{(*)}{\leq} \Delta^{-4t-32}~,
\end{align*}
where $(*)$ is satisfied if the exponent  $-(\Delta-2)^{t/2}\cdot\Delta^{-12}/12$ is smaller than $-(4t+32)\cdot \ln \Delta$ which holds for some $t=O(1)$ and implies an $r=O(1)$ that is independent from $v$ and the graph, in particular $r$ can be chosen independently from the graph size $n$.
\end{proof}

The following lemma is the most important result of the standard shattering technique. 
\begin{lemma}[The Shattering Lemma, \cite{FGLLL17}, cf. \cite{BEPSv3}]
\label{lem:shattering}
Let $H = (V,E)$ be a graph with maximum degree $\Delta$. Consider
a process which generates a random subset $B\subseteq V$ where $P(v \in B) \leq \Delta^{-c_1}$, for some constant $c_1 \geq 1$,
and that the random variables $1(v \in B)$ depend only on the randomness of nodes within at most $c_2$
hops from $v$, for all $v \in V$ , for some constant $c_2 \geq 1$. Moreover, let $Z = H[2c_2+1,4c_2+2]$ be the graph
which contains an edge between $u$ and $v$ iff their distance in $H$ is between $2c_2 + 1$ and $4c_2 + 2$. Let $L=H[B]$. Then
with probability at least $1 - n^{-c_3}$, for any constant $c_3$ satisfying $c_1>c_3+ 4c_2 + 2$, we have the following three
properties:
\begin{itemize}
\item [(P1)] $Z[B]$ has no connected component $U$ with $|U| \geq \log_{\Delta}n$.
\item[(P2)] Each connected component of $L$ has size at most $O(\log_{\Delta} n \cdot \Delta^{2c_2\Delta})$.
\item[(P3)] $L$ admits a $(\lambda,O(\log^{1/\lambda} n\cdot \log^2 \log n))$ network decomposition,
for any integer $\lambda \geq 1$, which can be computed by a randomized algorithm in $O\big(\lambda \log^{1/\lambda}n \cdot 2^{O(\sqrt{\log\log n})}\big)$ rounds, , w.h.p.
\item[(P4)] for any integer $R\geq 1$ there is a randomized algorithm to compute a $\big(2^{O(\sqrt{\log\log n})},R\cdot 2^{O(\sqrt{\log\log n})}\big)$ network decomposition of $L^R$ in $O(R\cdot 2^{O(\sqrt{\log\log n})})$ rounds, w.h.p.
\end{itemize}
\end{lemma}
\begin{proof}[Proof of \Cref{lem:shattering}]
$(P1)$-$(P3)$ are proven in \cite{FGLLL17}. The proof of $(P4)$ is along similar lines as the proof of $(P3)$ in \cite{FGLLL17} and we only provide a sketch here:
First one computes a ruling set $M$ with parameters $\big(2c_2+1, \Theta(\log\log n)\big)$ on $L^R$ with the randomized algorithm \Cref{lem:rulingSets}, (3). Similar to the arguments in \cite[Section 3.2, Step 3/4]{BEPSv3} this ruling set has, if restricted to a single connected component of $L$, at most $\log_{\Delta}n$ nodes. Now, we assign each node of the connected components to the closest ruling set node and form a cluster graph. Two clusters in this cluster graph are connected if they have two nodes that are neighbors in the original network. On this cluster graph and for each component in parallel we perform the deterministic network decomposition algorithm from \cite{panconesi1992improved} to compute a $\big(2^{O(\sqrt{\log N})}, 2^{O(\sqrt{\log N})}\big)$ for each cluster graph where $N=O(\log_{\Delta} n)$ is an upper bound on the size of each cluster graph. The runtime of the network decomposition depends on the size of the id space of the nodes and \cite[Remark 3.5]{BEPSv3} explains how to compute a new id space for each cluster graph. As one round on the cluster graph can be executed in $O\big(R\cdot\log\log n\big)$ rounds in $H$  the runtime of this step is $R\cdot 2^{O(\sqrt{\log N})}=R\cdot 2^{O(\sqrt{\log \log n})}$.
To obtain a network decomposition of $L^R$ we add each non ruling set node of $B$ to the cluster of its closest ruling set node. This increases the diameter of each cluster by at most a factor  $\Theta(\log\log n)$.
\end{proof}

\begin{remark}\label{rem:randomnessInSmallComponents}
The computation of the single network decomposition in $(P3)$ (or in $(P4)$) only uses randomness for the ruling set computation in the first step. In contrast to the deterministic network decomposition algorithm that is computed on each component separately and in parallel this randomized step is not performed on each component separately but on the whole graph. In particular its runtime and failure probability depend on $n$ where $n$ is the size of the original graph. Furthermore, the ruling set algorithm does not require that the components of size $N$ are also equipped with an ID space of size $\poly N$, but works with the ID space of the original graph. The same holds for the network decompositions and ruling sets that are computed to color the small components (cf. \Cref{sec:smallComponents}).
\end{remark}

\Cref{lem:shatteringProb,lem:shattering}  imply that the graph $L$ that remains after phase \ref{phase:assigningHappyNodes} consists of connected components of size at most $\poly \Delta \cdot\log_{\Delta}n$. \Cref{sec:smallComponents} explains in detail how these components can be $\Delta$-colored while respecting  the nodes colored with color \emph{one} in phase \ref{phase:marking}. 


\subsection{Shattering: Coloring Small Remaining Components (Phase \ref{phase:smallComponents})}
\label{sec:smallComponents}

We now explain how one can solve the small components that are left after the shattering process.
Let $C$ be  a small component with size at most $N:=\poly(\Delta)\cdot\log_{\Delta} n$. 
Call a node in $C$ \emph{free} if it has degree $<\Delta$ or at least one neighbor outside of $C$ that is not colored with the first color after the shattering process.
 We color the nodes of $C$ with the following algorithm where $R=2\log_{\Delta-2}N+1= O(\log \log n)$. The algorithm is explained from the view of a single component. 
\begin{enumerate}
\item Each free node selects itself. Further, each node that is contained in at least one DCC with radius at most $R$ selects one of these subgraphs. Let $\GraphCR$ be the virtual graph that has a node for each selected node and degree choosable subgraph. Any two subgraphs (or nodes)  of $\GraphCR$ are connected in $\GraphCR$ if they share a vertex or are connected by an edge in $G$. The maximum degree of $\GraphCR$ is $\min\{N,O(\Delta^{O(R)})\}$  and it has at most $|C|=N$ nodes. One round of an algorithm on $\GraphCR$ can be executed in $O(R)$ steps in $G$.

\item Find a $(2,\gamma)$ ruling set $M'$ of $\GraphCR$ where $\gamma=O(R)$ such that $\Delta(\GraphCR)^{2/\gamma}\leq \Delta^{1/2}$. 

		\runtimel We compute a $\big(2^{O(\sqrt{\log\log n})},4R\cdot 2^{O(\sqrt{\log\log n})}\big)$ network decomposition of $L^{4R}$ with \Cref{lem:shattering} (P4). Then each node assigns its color in this network decomposition to its corresponding selected node in $\GraphCR$. This yields a 
		$\big(2^{O(\sqrt{\log\log n})},4R\cdot2^{O(\sqrt{\log\log n})}\big)=
		\big(2^{O(\sqrt{\log\log n})},2^{O(\sqrt{\log\log n})}\big)$ network decomposition of  $\GraphCR$.  Then iterate through the colors of the network decomposition to compute the ruling set in time  $O\big(R\cdot 2^{O(\sqrt{\log\log n})}\big)=2^{O(\sqrt{\log\log n})}$.

		\runtimes  Use \Cref{lem:rulingSets}, (1)  in time $O\big(R\cdot\gamma\cdot\Delta(\GraphCR)^{2/\gamma}+\logstar n\big)=O\big(\log^2\log n\cdot\sqrt{\Delta}\big)$.
		
\item For $i=0,...,\gamma\cdot(R+1)+R$ define layers $D_i$ where $D_i$ consists of the nodes that are at distance $i$ to the closest node that is contained in a component in $M'$.

\runtimegeneral $O(R^2)=O(\log^2\log n)$.

\item \label{phase:smallLayerColoring} We color the layers in order $i=\gamma\cdot(R+1)+R, \ldots,1$; each layer is a  $\deg+1$ list coloring instance. There are $R^2+2R$ layers and we obtain the following runtimes.

		\runtimel In time $O((R^2+2R)\cdot 2^{O(\sqrt{\log \log n})})=2^{O(\sqrt{\log \log n})}$ via computing a single network decomposition for $C$ with \Cref{lem:shattering}, (P3). 

		\runtimes If we first use Linial's algorithm to compute a $O(\Delta^2)$ coloring the runtime is $O\big((R^2+2R)\cdot \sqrt{\Delta\log \Delta}\cdot\logstar \Delta\big)=O\big(\log^2\log n\sqrt{\Delta\log \Delta}\cdot\logstar \Delta\big)$ with \Cref{thm:listColoring}.

\item Now, we color the nodes that are in $D_0$. Each DCC is brute-forced independently in time $O(R)$. Each free node in $D_0$ can be colored in a single time unit as it has one uncolored neighbor outside the component it has a free color. 

\runtimegeneral $O(R)=O(\log\log n)$.
\end{enumerate}

\begin{lemma}\label{lem:claima1} If $D_0$ is not empty each node of the component is in one of the layers.\end{lemma}
\begin{proof}
The layers $D_0,\ldots, D_{\gamma\cdot (R+1)}$ contain all free nodes, all nodes that are in a DCC with radius at most $R$ and all nodes that have degree smaller $\Delta$. 
The layers $D_0,\ldots,D_{\gamma\cdot (R+1)+R}$ additionally contain the nodes that have such a DCC or such a node in distance at most $R$.
To show that all nodes are removed we assume that there is a node $v\in C$ that is in none of the layers. In particular it does not have a DCC or a free node in distance $R$, all nodes in its $R$-neighborhood have degree $\Delta$ or $\Delta-1$. As the $R$-neighborhood of $v$ does not contain a free node it can only hit the boundary of $C$ at colored nodes, i.e., its $R$-neighborhood can be obtained from the marking process as described in \Cref{ssec:structuralProperties}. Thus we can  apply \Cref{lem:detExpansion} and obtain that the BFS tree around $v$ and within the component expands and contains at least $(\Delta-2)^{R/2}> N$ nodes, a contradiction. 
\end{proof}

\begin{lemma}\label{lem:claimb1} $D_0$ is not empty. 

\end{lemma}
\begin{proof}
Assume that $D_0$ is empty.  
Let  $v$ be an arbitrary node of $C$.  Its $R$-neighborhood  neither contains a DCC of radius at most $R$ nor a free node  and all its nodes  have degree $\Delta$ or $\Delta-1$. As the $R$-neighborhood of $v$ does not contain a free node it can only hit the boundary of $C$ at colored nodes, i.e., its $R$-neighborhood can be obtained from the marking process as described in \Cref{ssec:structuralProperties}. Thus we can  apply \Cref{lem:detExpansion} and obtain that the BFS tree around $v$ and within the component expands and contains at least $(\Delta-2)^{R/2}> N$ nodes, a contradiction as in the worst case the whole component is a DCC (it cannot be an odd cycle due to $\Delta\geq 4$ and not a $(\Delta-1$)-clique; if it was a $(\Delta-1)$ clique and $D_0$ is empty all nodes have to be neighbors of the same marked  node (due to $b=6$) which implies a $\Delta$-clique, a contradiction).  
\end{proof}

The runtimes of the above algorithm provide the following lemma.
\begin{lemma}
\label{lem:smallComponents}
Let $\Delta\geq 4$. Then the small components can, w.h.p., be $\Delta$-colored in time

\[\min\left\{2^{O(\sqrt{\log\log n})} ,  O\left(\log^2\log n\cdot \sqrt{\Delta\log \Delta}\logstar\Delta\right)\right\}~.\] 
\end{lemma}
\begin{proof}
\Cref{lem:claima1,lem:claimb1} imply that each node is colored. The proof that coloring a single layer in phase \ref{phase:smallLayerColoring} is a a $\deg+1$ list coloring instance is along similar lines as in the proof of \Cref{thm:deterministicDelta}. 
The components and free nodes in $D_0$ can be colored independently because they stem from the independent set $M'$.
In both variants the runtime is dominated by phase \ref{phase:smallLayerColoring} step which implies the result.
\end{proof}
\begin{remark}
The algorithm to solve the small components only uses randomization to compute the network decomposition (\Cref{lem:shattering} and \Cref{rem:randomnessInSmallComponents}).
\end{remark}

\subsection{\texorpdfstring{Global Success After Marking Process for Small $\Delta$ (No Phase\ref{phase:smallComponents})}{Global Success After Marking Process for Small Delta (No Phase\ref{phase:smallComponents})}} \label{ssec:smallDelta}
In this section we show that a vertex $v\in V(H)$ is contained in $C_0,\ldots,C_{2r}, C'$ w.h.p. if $\Delta=O(1)$. 
As nodes which have an uncolored path of length $\leq 2r$ to a vertex of degree $<\Delta$ will be contained in one of the layers $C_0,\ldots,C_{2r}$ we assume throughout this section that $v$ and all vertices reachable from $v$ through uncolored path of length at most $r$ have degree $\Delta$ in $H$. 

Lemmas~\ref{lem:detExpansion} and~\ref{lem:detExpansionDelta3} imply that for $\Delta=O(1)$ and $b=15$ we can choose an $r = \Theta(\log \log n)$ such that for an arbitrarily large constant $c$, we have $|B_r(v)| \geq c\log_{\Delta} n$ after the marking process.

\begin{lemma} \label{lem:smallDeltaSingleSuccess}
Let $u$ be a node such that there is an unmarked, unselected path from $u$ to $v$. Then $u$ or its child $u'$ become a \Tnode of $v$ with a constant probability.
\end{lemma}
Note that the following analysis is done for $b = 15$. For $\Delta \geq 4$ we could equally well use $b = 6$ to optimize the constants.

\begin{proof}
	Let $u$ be a node such that there is an unmarked, unselected path from $v$ to $u$. In the following we consider the $2$-hop neighborhood of $u$ including marked nodes and naturally extend the BFS-tree from $v$ to $u$ for $2$ hops to $u$'s $2$-hop neighborhood. 
	We make a case distinction depending on how the children of $u$ in this BFS tree are connected. 
	
	\textit{Case 1, the children of $u$ do not form a clique:} 
  The children of $u$ form at least two distinct cliques (including single nodes). Among all pairs of non-adjacent neighbors of $u$ there are at most $\Delta-1$ pairs that include the parent of $u$ and at least $\Delta-2$ pairs that do not include the parent of $u$. Therefore, if $u$ is selected and does not back off due to another node being selected in distance at most $b$, it actually chooses a non-adjacent pair of neighbors that does not contain its parent, i.e., that does not block the uncolored path to $v$, with probability at least $(\Delta-2)/(\Delta-2+\Delta-1)\geq 1/3$. Thus
	node $u$ becomes a \Tnode for $v$ with probability at least $p' = (p/3)(1-p)^{\Delta^{15}} = \Theta(1)$, since $\Delta=O(1)$. 
	
	\textit{Case 2, the children of  $u$ form a clique:} In this case node $u$ cannot become a \Tnode of $v$ as it does not have two non-adjacent neighbors that do not block the path to $v$. However, $u$'s children must have a successor in the BFS tree, and therefore can become a \Tnode of $v$. We show that each child of $u$ becomes a \Tnode of $v$ with constant probability. Let $u'$ be one child of $u$ and $u''$ the unique child of $u'$. All children of $u$ form a clique and are connected to $u'$. If $u'$ is selected and does not back off the only pairs of non adjacent neighbors that $u'$ can select is $u''$ and $u$, or $u''$ and one of $u$'s children. In neither case the uncolored path to $v$ is blocked. With the same reasoning as before $u'$ becomes a \Tnode of $v$ with constant probability. 
	
	The events that $u$ or a child of $u$ succeed are not independent but are disjoint, so the claim holds.
\end{proof}

Note that the event in Lemma~\ref{lem:smallDeltaSingleSuccess} depends on randomness at distance at most 16.

\begin{lemma}
 The marking process generates a \Tnode for every node of the remainder graph $H$ with high probability.
\end{lemma}

\begin{proof}
  Consider an arbitrary node $v \in V(H)$. By Lemmas~\ref{lem:detExpansion} and~\ref{lem:detExpansionDelta3}, for any $c > 0$ we can choose $r = O(\log \log n)$ such that at distance $t$ from the root of any BFS tree, there are at least $c/p' \Delta^{16} \ln n$ nodes such that their path to $u$ is unmarked and unselected. From this set find a set $S$ of nodes as in Lemma~\ref{lem:smallDeltaSingleSuccess}, of size $c/p' \ln n$ with pairwise distance of at least 16 and they can each produce a \Tnode for $v$ with constant probability $p'$.

  The events that each $u \in S$ become a \Tnode of $v$ are independent due to the pairwise distance of the nodes. Thus no node of $S$ becomes a \Tnode of $v$ with probability at most
  \[
    (1-p')^{c/p' \ln n} \leq e^{-c \ln n} \leq n^{-c}~.
  \]
With a union bound over all nodes, all nodes of $H$ are happy with probability at least $1-1/n^{c-1}$.
\end{proof}

\section{Conclusion}\label{sec:conclusion}
We have provided several structural results for the $\Delta$-coloring (\Cref{sec:structural}) that hopefully will be of use for future algorithmic improvements to the problem. For constant degree graphs we provided a deterministic algorithm with $O(\log^2 n)$ round complexity (\Cref{thm:deterministicDelta}) and a $O(\log \log n)$ round randomized algorithm (\Cref{thm:mainDelta}) . The respective lower bounds are $\Omega(\log n)$ and $\Omega(\log \log n)$ and despite only a polynomial difference between upper and lower bounds it remains an intriguing open question whether the true complexity of the problem is at the lower or the higher end. 

After our submission, in a breakthrough result, Rozhon and Ghaffari \cite{RG19} have found a polylogarithmic deterministic time algorithm to compute $\big(\poly\log n, \poly\log n\big)$ network decompositions. As one of many implications the runtime of our deterministic $\Delta$-coloring algorithm (for unbounded degree) drops to $\poly\log n$ and our randomized algorithm for non-clique graphs with (unbounded) maximum degree $\Delta\geq 4$ drops to $\log \Delta + \poly\log\log n$ (\Cref{thm:main2}).

	\DeclareUrlCommand{\Doi}{\urlstyle{same}}
	\renewcommand{\doi}[1]{\href{http://dx.doi.org/#1}{\footnotesize\sf doi:\Doi{#1}}}

	\bibliographystyle{alpha}
	\bibliography{delta-col}

\newcommand{\etalchar}[1]{$^{#1}$}
\begin{thebibliography}{BFH{\etalchar{+}}16}

\bibitem[ABBE18]{aboulker18sparse}
Pierre Aboulker, Marthe Bonamy, Nicolas Bousquet, and Louis Esperet.
\newblock Distributed coloring in sparse graphs with fewer colors.
\newblock In {\em Proceedings of the 2018 {ACM} Symposium on Principles of
  Distributed Computing, {PODC} 2018, Egham, United Kingdom, July 23-27, 2018},
  pages 419--425, 2018.

\bibitem[ABI86]{alon1986fast}
Noga Alon, L{\'a}szl{\'o} Babai, and Alon Itai.
\newblock A fast and simple randomized parallel algorithm for the maximal
  independent set problem.
\newblock {\em Journal of Algorithms}, 7(4):567--583, 1986.

\bibitem[AGLP89]{awerbuch89}
B.~Awerbuch, A.~V. Goldberg, M.~Luby, and S.~A. Plotkin.
\newblock Network decomposition and locality in distributed computation.
\newblock In {\em Proc. 30th Symposium on Foundations of Computer Science (FOCS
  1989)}, pages 364--369, 1989.

\bibitem[BE13]{barenboim2013distributed}
Leonid Barenboim and Michael Elkin.
\newblock Distributed graph coloring: Fundamentals and recent developments.
\newblock {\em Synthesis Lectures on Distributed Computing Theory},
  4(1):1--171, 2013.

\bibitem[BEG18]{Elkin17}
Leonid Barenboim, Michael Elkin, and Uri Goldenberg.
\newblock Locally-iterative distributed ({$\Delta+1$})-coloring below
  szegedy-vishwanathan barrier, and applications to self-stabilization and to
  restricted-bandwidth models.
\newblock In {\em Proc. 37th ACM Symposium on Principles of Distributed
  Computing (PODC 2018)}, page To appear, 2018.

\bibitem[BEPS12]{barenboim2012locality}
Leonid Barenboim, Michael Elkin, Seth Pettie, and Johannes Schneider.
\newblock The locality of distributed symmetry breaking.
\newblock In {\em Proc. 53rd Symposium on Foundations of Computer Science (FOCS
  2012)}, pages 321--330, 2012.

\bibitem[BEPS16]{BEPSv3}
Leonid Barenboim, Michael Elkin, Seth Pettie, and Johannes Schneider.
\newblock The locality of distributed symmetry breaking.
\newblock {\em {Journal of the ACM}}, 63(3):20:1--20:45, 2016.

\bibitem[BFH{\etalchar{+}}16]{brandt2016LLL}
Sebastian Brandt, Orr Fischer, Juho Hirvonen, Barbara Keller, Tuomo
  Lempi{\"a}inen, Joel Rybicki, Jukka Suomela, and Jara Uitto.
\newblock A lower bound for the distributed lov{\'a}sz local lemma.
\newblock In {\em Proc. 48th ACM Symposium on Theory of Computing (STOC 2016)},
  pages 479--488. ACM, 2016.

\bibitem[BM76]{bondy1976graph}
J.A. Bondy and U.~Murty.
\newblock {\em Graph Theory with Applications}.
\newblock Elsevier, 1976.

\bibitem[Bro41]{brooks_1941}
R.~L. Brooks.
\newblock On colouring the nodes of a network.
\newblock {\em Mathematical Proceedings of the Cambridge Philosophical
  Society}, 37(2):194–197, 1941.

\bibitem[Bro09]{brooks2009colouring}
Rowland~Leonard Brooks.
\newblock On colouring the nodes of a network.
\newblock In {\em Classic papers in combinatorics}, pages 118--121. Springer,
  2009.

\bibitem[CKP16]{chang2016exponential}
Yi-Jun Chang, Tsvi Kopelowitz, and Seth Pettie.
\newblock An exponential separation between randomized and deterministic
  complexity in the local model.
\newblock In {\em Proc. 57th Symposium on Foundations of Computer Science (FOCS
  2016)}, pages 615--624. IEEE, 2016.

\bibitem[CLP18]{chang2017optimal}
Yi-Jun Chang, Wenzheng Li, and Seth Pettie.
\newblock An optimal distributed ({$\Delta+ 1$})-coloring algorithm?
\newblock In {\em Proc. 50th ACM Symposium on Theory of Computing (STOC 2018)},
  page To appear, 2018.

\bibitem[CM18]{chechik18}
Shiri Chechik and Doron Mukhtar.
\newblock Optimal distributed coloring algorithms for planar graphs in the
  {LOCAL} model.
\newblock {\em CoRR}, abs/1804.00137, 2018.

\bibitem[CR15]{cranston15brooks}
Daniel~W. Cranston and Landon Rabern.
\newblock Brooks' theorem and beyond.
\newblock {\em Journal of Graph Theory}, 80(3):199--225, 2015.

\bibitem[ERT79]{erdos79choosability}
Paul Erd\H{o}s, Arthur Rubin, and Herbert Taylor.
\newblock Choosability in graphs.
\newblock In {\em Proc. West Coast Conference on Combinatorics, Graph Theory
  and Computing}, volume~26, pages 125--157. Congressus Numerantium, 1979.

\bibitem[FG17]{FGLLL17}
Manuela Fischer and Mohsen Ghaffari.
\newblock Sublogarithmic distributed algorithms for lov{\'{a}}sz local lemma,
  and the complexity hierarchy.
\newblock In {\em Proc. 31st International Symposium on Distributed Computing
  (DISC 2017)}, volume~91 of {\em LIPIcs}, pages 18:1--18:16. Schloss Dagstuhl
  - Leibniz-Zentrum fuer Informatik, 2017.

\bibitem[FHK16]{fraigniaud15}
Pierre Fraigniaud, Marc Heinrich, and Adrian Kosowski.
\newblock Local conflict coloring.
\newblock In {\em Proc. 57th Symposium on Foundations of Computer Science (FOCS
  2016)}, pages 625--634, 2016.

\bibitem[Gha16]{ghaffari16improved}
Mohsen Ghaffari.
\newblock An improved distributed algorithm for maximal independent set.
\newblock In {\em Proc. 27th ACM-SIAM Symposium on Discrete Algorithms (SODA
  2016)}, pages 270--277, 2016.

\bibitem[GKM17]{SLOCAL17}
Mohsen Ghaffari, Fabian Kuhn, and Yannic Maus.
\newblock On the complexity of local distributed graph problems.
\newblock In {\em Proc. 49th ACM Symposium on Theory of Computing (STOC 2017)},
  pages 784--797. {ACM}, 2017.

\bibitem[GS17]{ghaffari2017splitting}
Mohsen Ghaffari and Hsin-Hao Su.
\newblock Distributed degree splitting, edge coloring, and orientations.
\newblock In {\em Proc. 28th ACM-SIAM Symposium on Discrete Algorithms (SODA
  2017)}, pages 2505--2523. Society for Industrial and Applied Mathematics,
  2017.

\bibitem[GV07]{Gfeller07}
Beat Gfeller and Elias Vicari.
\newblock A randomized distributed algorithm for the maximal independent set
  problem in growth-bounded graphs.
\newblock In {\em Proc. 26th ACM Symposium on Principles of Distributed
  Computing (PODC 2007)}, pages 53--60, New York, NY, USA, 2007. ACM.

\bibitem[HSS16]{harris2016distributed}
David~G Harris, Johannes Schneider, and Hsin-Hao Su.
\newblock Distributed (${\Delta}+ 1$)-coloring in sublogarithmic rounds.
\newblock In {\em Proc. 48th ACM Symposium on Theory of Computing (STOC 2016)},
  pages 465--478. ACM, 2016.

\bibitem[Lin92]{linial92}
Nati Linial.
\newblock Locality in distributed graph algorithms.
\newblock {\em SIAM Journal on Computing}, 21(1):193--201, 1992.

\bibitem[Lov75]{lovasz1975three}
L{\'a}szl{\'o} Lov{\'a}sz.
\newblock Three short proofs in graph theory.
\newblock {\em Journal of Combinatorial Theory, Series B}, 19(3):269--271,
  1975.

\bibitem[Lub86]{luby1986simple}
Michael Luby.
\newblock A simple parallel algorithm for the maximal independent set problem.
\newblock {\em SIAM Journal on Computing}, 15(4):1036--1053, 1986.

\bibitem[MR13]{molloy2013coloring}
Michael Molloy and Bruce Reed.
\newblock {\em Graph colouring and the probabilistic method}, volume~23.
\newblock Springer Science \& Business Media, 2013.

\bibitem[Pel00]{peleg2000distributed}
David Peleg.
\newblock Distributed computing.
\newblock {\em SIAM Monographs on discrete mathematics and applications}, 5,
  2000.

\bibitem[PS92]{panconesi1992improved}
Alessandro Panconesi and Aravind Srinivasan.
\newblock Improved distributed algorithms for coloring and network
  decomposition problems.
\newblock In {\em Proc. 24th ACM Symposium on Theory of Computing (STOC 1992)},
  pages 581--592. ACM, 1992.

\bibitem[PS95]{Panconesi1995}
Alessandro Panconesi and Aravind Srinivasan.
\newblock The local nature of {$\Delta$}-coloring and its algorithmic
  applications.
\newblock {\em Combinatorica}, 15(2):255--280, 1995.

\bibitem[RG19]{RG19}
V{\'{a}}clav Rozho\v{n} and Mohsen Ghaffari.
\newblock Polylogarithmic-time deterministic network decomposition and
  distributed derandomization.
\newblock {\em CoRR}, abs/1907.10937, 2019.

\bibitem[SEW13]{schneider13}
Johannes Schneider, Michael Elkin, and Roger Wattenhofer.
\newblock Symmetry breaking depending on the chromatic number or the
  neighborhood growth.
\newblock {\em Theoretical Computer Science}, 509:40--50, 2013.

\bibitem[Viz76]{vizing76vertex}
Vadim Vizing.
\newblock Vextex coloring with given colors.
\newblock {\em Metody Diskretn. Anal.}, pages 29:3--10, 1976.

\end{thebibliography}
\end{document}